\documentclass[11pt,a4paper]{article}

\usepackage{graphicx}
\graphicspath{{figures/}}
\usepackage[english]{babel}
\usepackage{latexsym}
\usepackage{url}
\usepackage{amsmath}
\usepackage{amssymb}
\usepackage{amsthm}
\usepackage{dsfont}
\usepackage{ifthen}
\usepackage{fancybox}
\usepackage{microtype,stmaryrd}
\usepackage{xspace}

\usepackage{color}
\definecolor{blu3}{rgb}{.1,.0,.4}
\usepackage{hyperref}
\hypersetup{colorlinks=true, linkcolor=blu3, urlcolor=blu3, citecolor=blu3, pdfpagemode=UseNone, pdfstartview=} 

\if10     
\usepackage[mathlines]{lineno}
\newcommand*\patchAmsMathEnvironmentForLineno[1]{%
  \expandafter\let\csname old#1\expandafter\endcsname\csname #1\endcsname
  \expandafter\let\csname oldend#1\expandafter\endcsname\csname end#1\endcsname
  \renewenvironment{#1}%
     {\linenomath\csname old#1\endcsname}%
     {\csname oldend#1\endcsname\endlinenomath}}%
\newcommand*\patchBothAmsMathEnvironmentsForLineno[1]{%
  \patchAmsMathEnvironmentForLineno{#1}%
  \patchAmsMathEnvironmentForLineno{#1*}}%
\AtBeginDocument{%
\patchBothAmsMathEnvironmentsForLineno{equation}%
\patchBothAmsMathEnvironmentsForLineno{align}%
\patchBothAmsMathEnvironmentsForLineno{flalign}%
\patchBothAmsMathEnvironmentsForLineno{alignat}%
\patchBothAmsMathEnvironmentsForLineno{gather}%
\patchBothAmsMathEnvironmentsForLineno{multline}%
}
\linenumbers
\fi

\usepackage[margin=1in]{geometry}

\newtheorem{theorem}{Theorem}

\newtheorem{lemma}[theorem]{Lemma}

\newcommand{\RR}{\ensuremath{\mathbb R}}  
\newcommand{\EE}{\ensuremath{\mathbb E}}  
\newcommand{\A}{\ensuremath{\mathcal A}}  
\newcommand{\I}{\ensuremath{\mathcal I}}  %

\DeclareMathOperator{\bb}{bb}
\DeclareMathOperator{\per}{per}
\DeclareMathOperator{\area}{area}

\DeclareMathOperator{\med}{med}
\DeclareMathOperator{\level}{level}
\DeclareMathOperator{\Out}{out} 
\newcommand{\CH}{\textit{CH\/}}

\DeclareMathOperator{\NE}{ne} 
\DeclareMathOperator{\NW}{nw} 
\DeclareMathOperator{\SE}{se} 
\DeclareMathOperator{\SNE}{NE} 
\DeclareMathOperator{\SNW}{NW} 
\DeclareMathOperator{\SSE}{SE} 

\newcommand{\V}[1]{\ensuremath{v_{\rm #1}}} 

\def\LengthInterval{{\sc LengthInterval}\xspace}
\def\Airport{{\sc Airport}\xspace}
\def\PerimeterBoundingBox{{\sc PerimeterBoundingBox}\xspace}
\def\PerimeterAnchoredBoundingBox{{\sc PerimeterAnchoredBoundingBox}\xspace}
\def\AreaAnchoredRectangles{{\sc AreaAnchoredRectangles}\xspace}
\def\AreaBand{{\sc AreaBand}\xspace}
\def\AreaBoundingBox{{\sc AreaBoundingBox}\xspace}
\def\AreaAnchoredBoundingBox{{\sc AreaAnchoredBoundingBox}\xspace}
\def\AreaEnclosingDisk{{\sc AreaEnclosingDisk}\xspace}
\def\PerimeterEnclosingDisk{{\sc PerimeterEnclosingDisk}\xspace}
\def\AreaConvexHull{{\sc AreaConvexHull}\xspace}
\def\PerimeterConvexHull{{\sc PerimeterConvexHull}\xspace}

\def\DEF#1{\textbf{\emph{#1}}}

\newcommand{\IGNORE}[1]{}

\begin{document}

\title{Computing Shapley values in the plane}

\author{Sergio Cabello\thanks{Department of Mathematics, IMFM, and
			Department of Mathematics, FMF, University of Ljubljana, Slovenia.
            Supported by the Slovenian Research Agency, program P1-0297 and projects J1-8130, J1-8155. 
			Email address: sergio.cabello@fmf.uni-lj.si.}
			\and
			Timothy M. Chan\thanks{Department of Computer Science,
			University of Illinois at Urbana-Champaign, USA.
			Email address: tmc@illinois.edu.}
}

\maketitle

\begin{abstract}
	We consider the problem of computing Shapley values for points in the plane,
	where each point is interpreted as a player, and the value of a coalition
	is defined by the area of usual geometric objects, such as the convex hull 
	or the minimum axis-parallel bounding box.
	
	For sets of $n$ points in the plane, we show how to compute 
	in roughly $O(n^{3/2})$ time the Shapley values for the area of  
	the minimum axis-parallel bounding box and the area of the union
	of the rectangles spanned by the origin and the input points. 
	When the points form an increasing or decreasing chain,
	the running time can be improved to near-linear.
	In all these cases, we use linearity of the Shapley values and
	algebraic methods.
	
	We also show that Shapley values for the area and the perimeter
	of the convex hull 
	or the minimum enclosing disk can be computed in $O(n^2)$ and $O(n^3)$ time, respectively. 
	These problems are closely related to the model of stochastic point sets 
	considered in computational geometry, 
	but here we have to consider random insertion orders of the points instead of 
	a probabilistic existence of points.
	
\end{abstract}

\section{Introduction}

One can associate several meaningful values to a set $P$ of points in the plane, like
for example the area of the convex hull or the area of the axis-parallel bounding box. 
How can we split this value among the points of $P$?
Shapley values are a standard tool in cooperative games to
``fairly" split common cost between different players. Our objective in this paper is
to present algorithms to compute the Shapley values for points in the plane when
the cost of each subset is defined by geometric means.

Formally, a \DEF{coalitional game} is a pair $(N,v)$, 
where $N$ is the set of players and 
$v\colon 2^N \rightarrow \RR$ is the \DEF{characteristic function}, 
which must satisfy $v(\emptyset)=0$. 
Depending on the problem at hand,
the characteristic function can be seen as a cost or a payoff associated to
each subset of players. 

In our setting, the players will be points in the plane, and we will use
$P\subset \RR^2$ for the set of players. 
Such scenario arises naturally in the context of game theory through modeling: 
each point represents an agent, 
and each coordinate of the point represents an attribute of the agent.  
We will consider characteristic functions defined through the area, 
such as the area of the convex hull.
As mentioned before, such area could be interpreted as a cost or a payoff associated
to each subset of the points $P$, depending on the problem.

Shapley values are probably the most popular solution concept 
for coalitional games, which in turn are a very common model for cooperative
games with transferable utility. The objective is to split the 
value $v(N)$ between the different players in a meaningful way.
It is difficult to overestimate the relevance of Shapley values.
See the book edited by Roth~\cite{Roth88} or the survey by Winter~\cite{Winter}
for a general discussion showing their relevance. 
There are also different axiomatic characterizations of the concept, meaning that Shapley
values can be shown to be the only map satisfying certain natural conditions.
Shapley values can be interpreted as a cost allocation, a split of the payoff,
or, after normalization, as a power index. The Shapley-Shubik power index
arises from considering voting games, a particular type of coalitional game.
In our context, the Shapley values that we compute can be interpreted 
as the relevance of each point within the point set for different geometric concepts.

In a nutshell, the Shapley value of a player $p$ is the expected increase in the 
value of the characteristic function 
when inserting the player $p$, if the players are inserted in a uniformly random order.
We provide a formal definition of Shapley values later on.
Since the definition is based on considering all the permutations of the players,
this way to compute Shapley values is computationally unfeasible.
In fact, there are several natural instances where computing Shapley values is difficult.
Thus, there is an interest to find scenarios where computing the Shapley values
is feasible. In this paper we provide some natural scenarios and discuss the
efficient computation.

The problems we consider here can be seen as natural extensions to $\RR^2$
of the classical airport problem considered by Littlechild and Owen~\cite{Littlechild}.
In the airport problem, we have a set $P$ of points with positive coordinate on the real line, 
and the cost of a subset $Q$ of the points is given by $\max(Q)$. It models the portion
of the runway that has to be used by each airplane, and Shapley values provide a way to split
the cost of the runway among the airplanes. 
As pointed out before, the points represent agents, in this case airplanes.
Several other airport problems are discussed in the survey by Thomson~\cite{Thomson13}.

The airport problem is closely related to the problem of allocating the length
of the smallest interval that contains a set of points on the line.
There are several natural generalizations of the intervals when we go from one to two dimensions,
and in this paper we consider several of those natural generalizations.

\paragraph{Coalitional games in the plane.}
The shapes that we consider in this paper are succinctly described in Figure~\ref{fig:intro}.
We next provide a summary of the main coalitional games that we consider;
other simpler auxiliary games are considered later.
In all cases, the set of players is a finite subset $P\subset \RR^2$.
\begin{description}
	\item[\AreaConvexHull game:]
		The characteristic function is 
		$\V{ch}(Q)=\area\left( CH(Q)\right)$ for each nonempty $Q\subset P$,
		where $CH(Q)$ denotes the convex hull of $Q$.
	\item[\AreaEnclosingDisk game:]
		The characteristic function is 
		$\V{ed}(Q)=\area\left( \med(Q)\right)$ for each nonempty $Q\subset P$,
		where $\med(Q)$ is a disk of smallest radius that contains $Q$.
	\item[\AreaAnchoredRectangles game:]
		The characteristic function is 
		$\V{ar}(Q)=\area\left( \bigcup_{p\in Q} R_p\right)$ for each nonempty $Q\subset P$,
		where $R_p$ is the axis-parallel rectangle with one corner at $p$
		and another corner at the origin.
	\item[\AreaBoundingBox game:]
		The characteristic function is 
		$\V{bb}(Q)=\area\left( \bb(Q)\right)$ for each nonempty $Q\subset P$,
		where $\bb(Q)$ is the smallest axis-parallel bounding box of $Q$.
	\item[\AreaAnchoredBoundingBox game:]
		The characteristic function is 
		$\V{abb}(Q)=\area\left( \bb(Q\cup \{o \})\right)$ for each nonempty $Q\subset P$,
		where $o$ is the origin.
\end{description}
In our discussion we focus on the area of the shapes. 
One can consider the variants where the perimeter of the shapes is used.
The name of the problem is obtained by replacing {\sc Area$\star$} by {\sc Perimeter$\star$}.
Handling the perimeter is either substantially easier or can be solved
by similar methods, and it will be discussed along the way.

Note that in all the problems we consider monotone characteristic functions: 
whenever $Q\subset Q'\subset P$ we have $v(Q)\le v(Q')$. 

\begin{figure}
	\centering
	\includegraphics[width=\textwidth]{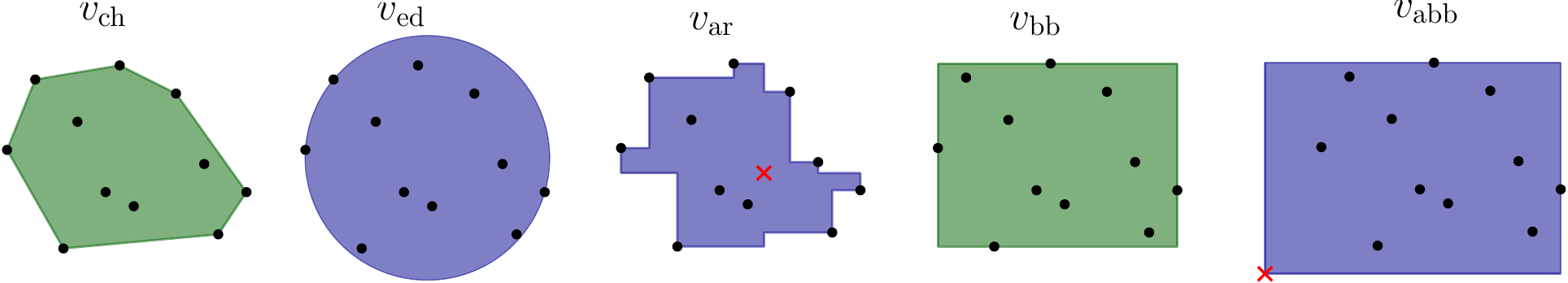}
	\caption{Different costs associated to a point set that are considered in this paper.
		The cross represents the origin. In all cases we focus on the area.}
	\label{fig:intro}
\end{figure}

\paragraph{Overview of our contribution.}
We show that the Shapley values for the
\AreaConvexHull and \AreaEnclosingDisk games 
can be computed in $O(n^2)$ and $O(n^3)$ time, respectively. 
These problems resemble the models recently considered in stochastic 
computational geometry; 
see for example~\cite{AgarwalHSYZ17,AgarwalKSS18,FinkHKS17,KamousiCS11,KamousiCS14,Perez-Lantero16,XueLJ18}. 
However, there are some key differences in the models.
In the most basic model in stochastic computational geometry, sometimes called
\emph{unipoint model},
we have a point set and, for each point,
a known probability of being actually present, independently for each point. 
Then we want to analyze a certain
functional, like the expected area of the minimum enclosing disk, the expected
area of the smallest bounding box, or the probability that some point is contained
in the convex hull. 

In our scenario, we have to consider random insertion orders of the points
and analyze the expected increase in the value of the characteristic function after the insertion 
of a fixed point $p$.
Thus, we have to consider subsets of points constructed according to a different random 
process. In particular, whether other points precede $p$ or not in the random order
are not independent events.
In our analysis, we condition on properties of the shape before adding the new point $p$.
In the case of the minimum enclosing disk we use that each minimum enclosing disk 
is determined by at most $3$ points,
while in the case of the convex hull we condition separately on each single edge of the convex hull 
before the insertion of $p$.
A straightforward application of this principle gives polynomial-time algorithms.
To improve the running time we carry out this idea but split it into two computations:
the expected value after the insertion of $p$, if the insertion of $p$ changed the shape;
the expected value before the insertion of $p$, if the insertion of $p$ changed the shape.
Finally, we use arrangements of lines and planes to speed up the computation by an additional
linear factor. 

For the 
\AreaAnchoredRectangles,
\AreaBoundingBox, and
\AreaAnchoredBoundingBox games
we show that Shapley values can be computed in $O(n^{3/2}\log n)$ time.
In the special case where the points form a \emph{chain} (increasing or decreasing $y$-coordinate
for increasing $x$-coordinate), the Shapley values of those games can
be computed in $O(n\log^2 n)$ time.
We refer to these games as \emph{axis-parallel} games.

It is relative easy to compute the Shapley values for these axis-parallel games 
in quadratic time using arrangements of rectangles and the linearity of Shapley values.
We will discuss this as an intermediary step towards our solution.
However, it is not obvious how to get subquadratic time.
Besides using the linearity of Shapley values, a key ingredient in our algorithms 
is using convolutions to evaluate at multiple points some special rational functions 
that keep track of the ratio of permutations with a certain property. 
The use of computer algebra in computational geometry
is very recent, and there are very few results~\cite{AjwaniRST07,AronovK18,MorozA16} 
using such techniques in geometric problems. (In contrast, there is a quickly growing
body of works in \emph{discrete} geometry using real algebra.)

At the basement of our algorithms, we need to count the number of permutations with 
certain properties. For this we use some simple combinatorial counting. 
In summary, our results combine fundamental concepts from several different areas
and, motivated by classical concepts of game theory,
we introduce new problems related to stochastic computational geometry
and provide efficient algorithms for them.

\paragraph{Related work.}
The chapter by Deng and Fang~\cite{Deng2008} gives a summary of computational 
aspects of coalitional games. The book by Nisan et al.~\cite{NisanRTV07} provides
a general overview of the interactions between game theory and algorithms.

Puerto, Tamir and Perea~\cite{PuertoTP11} study the the Minimum Radius Location Game
in general metric spaces. When specialized to the Euclidean plane, this is equivalent
to the \PerimeterEnclosingDisk. The paper also considers the $L_1$-metric,
which is in some way related to the axis-parallel problems we consider here.
However, the focus of their research is on understanding the core of the game,
and do not discuss the computation of Shapley values.
In particular, they show that the Minimum Radius Location Game in the Euclidean plane 
has nonempty core.
Puerto, Tamir and Perea~\cite{PuertoTP12} also discuss the Minimum Diameter Location Game,
which can be defined for arbitrary metric spaces, but then focus their discussion on graphs. 

Faigle et al.~\cite{Faigle1998} consider the TSP coalitional game in general metric
spaces and specialize some results to the Euclidean plane. They also analyze an approximate
allocation of the costs. 

The computation of Shapley values have been considered for several games on graphs.
The aforementioned \Airport problem can be considered a shortest spanning-path game
in a (graph-theoretic) path.
Megiddo~\cite{Megiddo78} extended this to trees, while
Deng and Papadimitriou~\cite{DengP94} discuss a game on arbitrary graphs defined
by induced subgraphs.
There is a very large body of follow up works for graphs, but we could not trace
other works considering the computation of Shapley values for games defined
through planar objects, despite being very natural.

\paragraph{Assumptions.}
We will assume general position in the following sense: no two points
have the same $x$ or $y$ coordinate, no three points are collinear, 
and no four points are cocircular.
In particular, the points are all different.
The actual assumptions depend on the game under consideration.
It is simple to consider the general case, but it makes the notation more tedious.

We assume a unit-cost real-RAM model of computation. In a model of computation that accounts for bit complexity, time bounds may increase by polynomial factors (even if the input numbers are integers, the outputs may be rationals with large numerators and denominators).

\paragraph{Organization.}
We start with preliminaries in Section~\ref{sec:preliminaries}, where
we set the notation, define Shapley values, explain some basic consequences of
the \Airport game, discuss our needs of algebraic computations, and 
count permutations with some properties.
The section is long because of our use of tools from different areas.

Then we analyze different games.
The \AreaConvexHull and \AreaEnclosingDisk games are considered in 
Section~\ref{sec:CH} and Section~\ref{sec:enclosingdisk}, respectively.
In Section~\ref{sec:anchoredrectangles} we discuss the \AreaAnchoredRectangles game,
while in Section~\ref{sec:boundingbox} we discuss the \AreaBoundingBox
and \AreaAnchoredBoundingBox games.
To understand Section~\ref{sec:boundingbox} one has to go through 
Section~\ref{sec:anchoredrectangles} before. The remaining sections (and games)
have no dependencies between them and can be read in any order.
We conclude in Section~\ref{sec:conclusions} with some discussion.

\section{Preliminaries}
\label{sec:preliminaries}
Here we provide notation and background used through the paper.

\paragraph{Geometry.}
In most cases we will consider points in the Euclidean plane $\RR^2$.
The origin is denoted by $o$.
For a point $p\in \RR^2$, let $R_p$ be the axis-parallel rectangle with one corner at the
origin $o$ and the opposite corner at $p$.
As already mentioned in the introduction, for each $Q\subset \RR^2$
we use $\bb(Q)$ for the (minimum) axis-parallel bounding box of $Q$,
$\med(Q)$ for a smallest (actually, \emph{the} smallest) disk that contains $Q$,
and $CH(Q)$ for the convex hull of $Q$.

We try to use the word \emph{rectangle} when one corner is defined by an input point,
while the word \emph{box} is used for more general cases. 
All rectangles and boxes in this paper are axis-parallel, and we drop the adjective
\emph{axis-parallel} when referring to them.
An \emph{anchored} box is a box with one corner at the origin.

For a point $p\in \RR^2$ we denote by $x(p)$ and $y(p)$ its $x$- and $y$-coordinate, respectively.
We use $x(Q)$ for $\{ x(q)\mid q\in Q\}$, and similarly $y(Q)$.
A set $P$ of points is a \DEF{decreasing chain}, if $x(p)<x(q)$
implies $y(p)>y(q)$ for all $p,q\in P$.
A set $P$ of points is an \DEF{increasing chain} if $x(p)<x(q)$
implies $y(p)<y(q)$ for all $p,q\in P$.

\subsection{Shapley values.}
We provide a condensed overview suitable for understanding our work.
We refer to some textbooks in Game Theory (\cite[Chapter IV.3]{Ferguson},
\cite[Section~9.4]{Myerson}, \cite[Section 14.4]{OsborneRubinstein})
for a comprehensive treatment that includes a discussion of
the properties of Shapley values.

Consider a coalitional game $(P,v)$.
For us, a player is usually denoted by $p\in P$. 
We denote by $n$ the number of players and by $[n]$ the 
set of integers $\{ 1,\dots, n\}$.

A permutation of $P$ is a bijective map $\pi\colon P\rightarrow [n]$.
Let $\Pi(P)$ be the set of permutations of $P$.
Each permutation $\pi\in \Pi(P)$ defines an ordering in $P$,
where $\pi(p)$ is the position of $p$ in that order.
We will heavily use this interpretation of permutations as defining an order in $P$.
For each element $p\in P$ and each permutation $\pi\in \Pi(P)$, 
let $P(\pi,p)$ be the elements of $P$ before $p$ in the order defined by $\pi$.
Thus
\[
	P(\pi,p)~=~ \{ q\in P\mid \pi(q)\le \pi (p)\}.
\]
We can visualize $P(\pi,p)$ as adding the elements of $P$ one by one,
following the order defined by $\pi$, until we insert $p$.
The increment in $v(\cdot)$ when adding player $p$ is
\[
	\Delta(v,\pi,p)~=~ v\bigl( P(\pi,p) \bigr) - v\bigl(P(\pi,p)\setminus \{ p \}\bigr).
\]
The \DEF{Shapley value} of player $p\in P$ in the game $(P,v)$ is 
\[
	\phi(p,v)~=~ \frac{1}{n!}~ \sum_{\pi\in \Pi(P)} \Delta(v,\pi,p). 
\]
It is straightforward to note that
\[
	\phi(p,v)~=~ \EE_\pi \Bigl[ \Delta(v,\pi,p) \Bigr],
\]
where $\pi$ is picked uniformly at random from $\Pi(P)$.
This means that Shapley values describe the expected increase in $v(\cdot)$
when we insert the player $p$ in a random permutation.
 
Since several permutations $\pi$ define the same subset $P(\pi,p)$,
the Shapley value of $p$ is often rewritten as
\[
	\phi(p,v)~=~ \sum_{S\subset P\setminus\{ p \} } 
		\frac{|S|! (n-|S|-1)!}{n!} \Bigl( v(S\cup \{p\})- v(S) \Bigr).
\]
We will use a similar principle in our algorithms: 
find groups of permutations that contribute the same to the sums.
For non-axis-parallel games we will also use other equivalent expressions of
the Shapley value that follow from rewriting the sum differently.

It is easy to see that Shapley values are linear in the characteristic functions.
Indeed, since for any two characteristic functions 
$v_1,v_2\colon 2^P \rightarrow \RR$ and for each $\lambda_1,\lambda_2\in \RR$ we have
\begin{align*}
	\Delta(\lambda_1 v_1+\lambda_2 v_2,\pi,p) ~=~
	\lambda_1\cdot \Delta(v_1,\pi,p)+\lambda_2\cdot \Delta(v_2,\pi,p)
\end{align*}
we obtain
\begin{align*}
	&\phi(p,\lambda_1v_1+\lambda_2v_2) ~=~ \lambda_1\cdot\phi(p,v_1)+ \lambda_2\cdot\phi(p,v_2).
\end{align*}
In fact, (a weakened version of) linearity is one of the properties  
considered when defining Shapley values axiomatically.

It is easy to obtain the Shapley values when the characteristic function $v$ is constant
over all nonempty subsets of $P$. In this case $\phi(p,v)=v(P)/n$.

We say that two games $(P,v)$ and $(P',v')$, where $P$ and $P'$ are 
point sets in Euclidean space, are \DEF{isometrically equivalent} 
if there is some isometry that transforms one into the other.
That is, there is some isometry
$\rho\colon \RR^2 \rightarrow \RR^2$ such that $P'=\rho(P)$ and $v'=v\circ \rho$.
Finding Shapley values for $(P,v)$ or $(P',v')$ is equivalent
because $\phi(\rho(p),v')=\phi(p,v)$ for each $p\in P$.

\subsection{Airport and related games.}

We will consider the following games, some of them $1$-dimensional games.

\begin{description}
	\item[\Airport game:]
		The set of players $P$ is on the positive side of the real line and
		the characteristic function is 
		$\V{air}(Q)=\max(Q)$ for each nonempty $Q\subset P$.
	\item[\LengthInterval game:]
		The set of players $P$ is on the real line and
		the characteristic function is 
		$\V{li}(Q)=\max(Q) - \min(Q)$ for each nonempty $Q\subset P$.
	\item[\AreaBand game:] 
		The set of players $P$ is on the plane 
		and the characteristic function is 
		$\V{ab}(Q)=\left(\max(y(P))-\min(y(P)\right)\cdot \V{li}(x(Q)$ 
		for each nonempty $Q\subset P$.
	\item[\PerimeterBoundingBox game:] 
		The set of players $P$ is on the plane 
		and the characteristic function is 
		$\V{pbb}(Q)=\per(\bb(Q))$ for each nonempty $Q\subset P$.
	\item[\PerimeterAnchoredBoundingBox game:] 
		The set of players $P$ is on the plane 
		and the characteristic function is 
		$\V{pabb}(Q)=\per(\bb(Q\cup\{ o\}))$ for each nonempty $Q\subset P$.
\end{description}

\begin{lemma}
\label{le:basicgames}
	The \Airport, \LengthInterval, \AreaBand games, \PerimeterBoundingBox and
	\PerimeterAnchoredBoundingBox games can be solved in $O(n\log n)$ time.
\end{lemma}
\begin{proof}
	The \Airport game was solved by Littlechild and Owen~\cite{Littlechild}, as follows.
	If $P$ is given by $0 < x_1<\dots< x_n$, then  
	\[
		\phi(x_i,\V{air})= 
		\begin{cases}
			\tfrac{x_1}{n} & \text{if $i=1$,}\\
			\phi(x_{i-1},\V{air})+\tfrac{x_i-x_{i-1}}{n-i+1}& \text{for $i=2,\dots, n$.}
		\end{cases}
	\]
	This shows the result for the \Airport game because it requires sorting.
	
	The \LengthInterval game can be reduced to the \Airport problem using simple inclusion-exclusion on the line.
	Consider the values $\alpha=\min(P)$ and $\beta=\max(P)$
	and define the characteristic functions		
	\begin{align*}
		v_1(Q) ~&=~ \max(Q)-\alpha \text{ for } Q\neq \emptyset;\\
		v_2(Q) ~&=~ \beta - \min(Q) \text{ for } Q\neq \emptyset;\\
		v_3(Q) ~&=~ \alpha-\beta 	 \text{ for } Q\neq \emptyset.
	\end{align*}
	They define games that are isometrically equivalent to \Airport games or a constant game,
	and we have
	$\V{li}(Q)=v_1(Q)+v_2(Q)+v_3(Q)$ for all nonempty $Q\subset P$. 
	The result follows from the linearity of the Shapley values.
	
	For the \AreaBand game we just notice that 
	$\V{ab}(Q)=\V{li}(Q)\cdot (\max(y(P))-\min(y(P)))$
	and use again the linearity of Shapley values. 
	(The value $\max(y(P))-\min(y(P))$ is constant.)
	For the \PerimeterBoundingBox game we note that
	\[
		\V{pbb}(Q)= 2\cdot (\max(x(P))-\min(x(P)) )+ 2\cdot (\max(y(P))-\min(y(P)) ),
	\]
	which means that we need to compute the Shapley values for two \LengthInterval games.
	For the \PerimeterAnchoredBoundingBox we use that
	\[
		\V{pabb}(Q)= 2\cdot (\max(x(P),0)-\min(x(P),0) )+ 2\cdot (\max(y(P),0)-\min(y(P),0) ),
	\]
	which is a linear combination of a few \Airport games.
\end{proof}
 
The perimeter for the anchored bounding box is the same as for the
union of anchored rectangles. So this solves the perimeter games for all
axis-parallel objects considered in the paper.

\noindent\textbf{Remark.} 
We can now point out the convenience of using the real-RAM model of computation.
In the \Airport game with $x_i=i$ for each $i\in [n]$, 
we have $\phi(x_n,\V{air})=\sum_{i\in [n]} \tfrac{1}{i}$, the $n$-th harmonic number.
Thus, even for simple games it would become a challenge to carry precise bounds on
the number of bits.

\subsection{Algebraic computations.}
For the axis-parallel problems we will be using the fast Fourier transform to compute convolutions.
Assume we are given values $a_0,\dots,a_n$ and $b_0,\dots,b_n$, and define $a_i$ and $b_i$
equal zero for all other indices.
Using the fast Fourier transform we can compute the convolutions $c_k = \sum_{i+j=k} a_ib_j$
for all integer $k$ in $O(n \log n)$ operations. (The value is nonzero for at most $2n+1$ indices.)
Our use of this is encoded in the following lemma, which provides multipoint evaluation 
for a special type of rational functions. Note that using the more generic
approach of Aronov, Katz and Moroz~\cite{AronovK18,MorozA16} for multipoint rational functions 
gives a slightly worse running time.

\begin{lemma}
\label{le:multipoint_evaluation}
	Let $b_0,\dots, b_n, \Delta$ be real numbers
	and consider the rational function
	\[
		R(x) ~=~ \sum_{t=0}^n \frac{b_t}{\Delta+t+x}.
	\]
	Given an integer $\ell> -\Delta$, possibly negative, and a positive integer $m$,
	we can evaluate $R(x)$ at all the integer values $x= \ell, \ell+1, \dots, \ell+m$
	in $O((n + m) \log(n+m))$ time.
\end{lemma}
\begin{proof}
	Set $a_i= 1/(\Delta+\ell+m-i)$ for all $i\in \{ 0,\dots, m\}$.
	Note that the assumption $\Delta+\ell> 0$ implies that the denominator is always positive.
	All the other values of $a_*$ and $b_*$ are set to $0$.
	Define the convolutions $c_k = \sum_{i} a_ib_{k-j}$ for all $k\in \{0,\dots,m \}$
	and compute them using the fast Fourier transform in $O((n+m)\log (n+m))$ time.
	Then, for each integer $x$ we have
	\[
		R(x) ~=~ \sum_{t=0}^n \frac{b_t}{\Delta+t+x} 
			 ~=~ \sum_{t=0}^n b_t a_{m+\ell-t-x}
			 ~=~ c_{m+\ell-x}.
	\]
	Thus, from computing the convolution of $a_*$ and $b_*$, we 
	get the values $R(\ell),\dots, R(\ell+m)$
\end{proof}

\subsection{Permutations.}
We will have to count permutations with certain properties.
The following lemmas encode this.

\begin{lemma}
\label{le:permutations2}
	Let $N$ be a set with $n$ elements. Let $\{ x \}$, $A$, and $B$
	be disjoint subsets of $N$ and set $\alpha=|A|$, $\beta=|B|$.
	There are $n!\cdot \frac{\alpha! \beta!}{(\alpha+\beta+1)!}$ permutations of $N$ such that  
	all the elements of $B$ are before $x$ and all the elements of $A$ are after $x$.
\end{lemma}
\begin{proof}
	Permute the $\alpha$ elements of $A$ arbitrarily and put them after $x$.
	Similarly, permute the $\beta$ elements of $B$ arbitrarily and put them before $x$.
	Finally, insert the remaining $|N\setminus (A\cup B\cup \{ x\})|=n-\alpha-\beta-1$
	elements arbitrarily, one by one.
	Each permutation that we want to count is constructed exactly once
	with the procedure, and thus there are
	\[
		\alpha!\cdot \beta!\cdot \Bigl( (\alpha+\beta+2)\cdots n \Bigr) ~=~ 
		\frac{n! \alpha! \beta!}{(\alpha+\beta+1)!}. 
	\qedhere
	\]
\end{proof}

\begin{lemma}
\label{le:permutations1}
	Let $N$ be a set with $n$ elements.
	Let $A$, $B$ and $C$ be pairwise disjoint subsets of $N$ with 
	$\alpha$, $\beta$ and $\gamma$ elements, respectively.
	\begin{itemize}
		\item Consider a fixed element $a\in A$. There are 
			$n! \cdot \left(\frac{1}{\alpha+\beta}+ \frac{1}{\alpha+\gamma}-
				\frac{1}{\alpha+\beta+\gamma}\right)$
			permutations of $N$ such that:
			\textup{(i)} $a$ is the first element of $A$ and 
			\textup{(ii)} $a$ is before all elements of $B$ or before all elements of $C$.
			(Thus, the whole $B$ or the whole $C$ is behind $a$.)
		\item Consider a fixed element $b\in B$. There are 
			$n! \cdot \left( \frac{1}{\alpha+\beta}- \frac{1}{\alpha+\beta+\gamma}\right) $
			permutations of $N$ such that: 
			\textup{(i)} $b$ is the first element of $B$, 
			\textup{(ii)} $b$ is before all elements of $A$, and
			\textup{(iii)} at least one element of $C$ is before $b$.
	\end{itemize}
\end{lemma}
\begin{proof}
	We start showing the first item. Fix an element $a\in A$.
	For a set $X$, disjoint from $A$, let $\Pi_X$ be the set of permutations that have $a$ 
	before all elements of $(A\setminus \{ a \})\cup X$.
	We can count these permutations as follows.
	First, permute the elements of $(A\setminus \{ a \})\cup X$ and fix their order. 
	Then insert $a$ at the front.
	Finally, insert the elements of $N\setminus (A\cup X)$ anywhere, one by one. 
	All permutations of $\Pi_X$ are produced with this procedure exactly once.
	Therefore, there are
	\[
		(\alpha+|X|-1)! \cdot \Bigl( (\alpha+|X|+1)\cdot (\alpha+|X|+2)\cdots n \Bigr) 
		~=~ n! \frac{1}{\alpha+|X|}
	\]
	permutations in $\Pi_X$.
	Using inclusion-exclusion we have
	\begin{align*}
		|\Pi_B\cup \Pi_C| ~&=~ |\Pi_B|+|\Pi_C| - |\Pi_B\cap \Pi_C|\\
			&=~ |\Pi_B|+|\Pi_C| - |\Pi_{B\cup C}| \\
			&=~ n! \cdot \left(\frac{1}{\alpha+\beta}+ \frac{1}{\alpha+\gamma}- \frac{1}{\alpha+\beta+\gamma}\right).
	\end{align*}
	This finishes the proof of the first item.
	
	Now we prove the second item. Fix an element $b\in B$. 
	We construct the desired permutations as follows.
	Select one element $c\in C$ and place $c$ before $b$.
	Now insert all the elements of $A\cup B\setminus \{ b \}$ after $b$ in arbitrary order.
	Then, insert all the elements of $C\setminus \{ c\}$ after $c$, one by one.
	Finally, insert the remaining elements in any place, one by one.
	Each of permutations under consideration is created exactly once by this procedure.
	Therefore, there are 
	\[
		\gamma \cdot (\alpha+\beta-1 )! \cdot 
		\Bigl( (\alpha+\beta+1)\cdots (\alpha+\beta+\gamma-1) \Bigr) \cdot 
		\Bigl( (\alpha+\beta+\gamma+1)\cdots n \Bigr)
	\]
	permutations, which is exactly $n!\cdot \frac{\gamma}{(\alpha+\beta)(\alpha+\beta+\gamma)}$.
	Then we use that
	\[
		\frac{1}{\alpha+\beta}- \frac{1}{\alpha+\beta+\gamma} ~=~ 
		\frac{\gamma}{(\alpha+\beta)(\alpha+\beta+\gamma)}. 
		\qedhere
	\]
\end{proof}

\section{Convex hull}
\label{sec:CH}

In this section we consider the area and the perimeter
of the convex hull of the points. We will focus on the area, thus the characteristic function $\V{ch}$, 
and at the end we explain the small changes needed to handle the perimeter.
Consider a fixed set $P$ of points in the plane.
For simplicity we assume that no three points are collinear.

\begin{lemma}
\label{le:convex_hull}
	For each point $p$ of $P$ we can compute $\phi(p,\V{ch})$ in $O(n^2)$ time.
\end{lemma}
\begin{proof}
	For each $q,q'\in P\ (q\neq q')$, 
	let $H(q,q')$ be the open halfplane containing all points to the left of 
	the directed line from $q$ to $q'$.
	Define $\level(q,q')$ to be the number of points in $P\cap H(q,q')$.

 	\begin{figure}
		\centering
		\includegraphics[width=\textwidth,page=4]{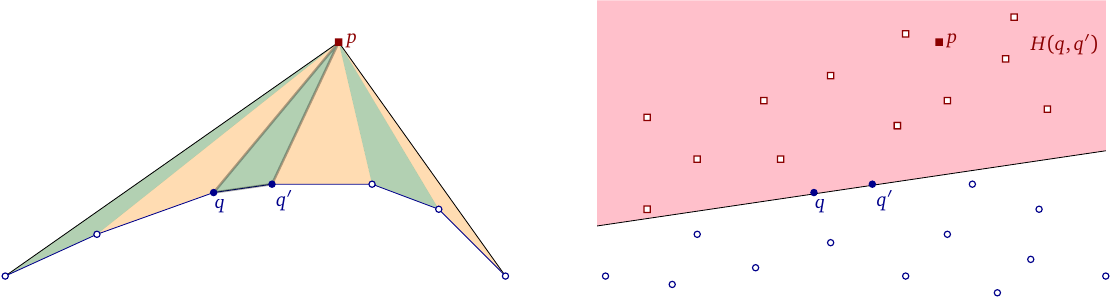}
		\caption{Left: triangulating the difference between  $CH(P(\pi,p))$ and $CH(P(\pi,p)\setminus \{p\})$. 
		Right: in order for $\triangle pqq'$ to be in $T(\pi,p)$,
		$q$ and $q'$ must appear before $p$, which in turn must appear
		before all other points in the halfplane $H(q,q')$.}
		\label{fig:convex_hull4}
	\end{figure}
	
	We can decompose the difference between
	$\CH(P(\pi,p))$ and $\CH(P(\pi,p)\setminus\{p\})$
	into a set $T(\pi,p)$ of triangles (see Figure~\ref{fig:convex_hull4} (left)),
	where 
  	\begin{align*}
		T(\pi,p) ~=~ \{ \triangle pqq'\mid\ & p\in H(q,q'),\\
		& \text{$q$ and $q'$ appear before $p$ in $\pi$, and}\\
		& \text{no points before $p$ in $\pi$ lie in $H(q,q')$} \}.
	\end{align*}
	In other words, $\triangle pqq'\in T(\pi,p)$ if and only if
	$p\in H(q,q')$, and among the $\level(q,q')+2$ points in $ H(q,q')\cup\{q,q'\}$,
	the two earliest points are $q$ and $q'$, and the third earliest
	point is $p$.  See Figure~\ref{fig:convex_hull4} (right).  
	(Note that if $\CH(P(\pi,p))=\CH(P(\pi,p)\setminus\{p\})$, then $T(\pi,p)$ is empty.)
	For fixed $p,q,q'\in P$ with $p\in H(q,q')$, Lemma~\ref{le:permutations2}
	with $x=p$, $B=\{q,q'\}$ and $A=H(q,q')\cap P\setminus \{ p\}$ tells
	that the probability that 
	$\triangle pqq'\in T(\pi,p)$ with respect to a random permutation $\pi$
	is exactly 
	\[ 
		\rho(q,q') ~=~ \frac{(\level(q,q')-1)!\, 2! }{(\level(q,q')+2)!} 
				   ~=~ \frac{2}{(\level(q,q')+2)(\level(q,q')+1)\level(q,q')}.
	\] 
	It follows that the Shapley value of $p$ is 
	\begin{equation}\label{eqn:ch}
		\phi(p,\V{ch})\ = \sum_{\mbox{\scriptsize\begin{tabular}{c}$q,q'\in P\ (q\neq q')$\\ with $p\in H(q,q')$\end{tabular}}} 
			\area(\triangle pqq')\cdot\rho(q,q').
	\end{equation}
	From the formula, we can immediately compute $\phi(p,\V{ch})$ 
	for any given $p\in P$ in $O(n^2)$ time, if all $\rho(q,q')$ values have been precomputed.

	Each value $\rho(q,q')$ can be computed from $\level(q,q')$	using $O(1)$ arithmetic operations.
	Thus, precomputing $\rho(q,q')$ requires precomputing $\level(q,q')$ for all $O(n^2)$ pairs $q,q'$.
	In the dual, this corresponds to computing the 
	\emph{levels} of all $O(n^2)$ vertices in an arrangement of $n$ lines.  The arrangement of $n$ lines can be
	constructed in $O(n^2)$ time~\cite[Chapter 8]{BergCKO08}, and the levels of all vertices can be
	subsequently generated by traversing the arrangement in $O(1)$ time per vertex.
\end{proof}

Naively applying Lemma~\ref{le:convex_hull} to all points $p\in P$ gives
$O(n^3)$ total time.  We can speed up the algorithm by a factor of $n$:

\begin{theorem}
\label{thm:ach}
	The Shapley values of the \AreaConvexHull game for $n$ points can be computed 
	in $O(n^2)$ time. 	
\end{theorem}
\begin{proof}
	Let $p=(x,y)\in P$.
	Observe that for fixed $q,q'\in P\ (q\neq q')$, if
	$p\in H(q,q')$, then 
	$\area(\triangle pqq')$ is a linear function in $x$ and $y$ and
	can thus be written as $a(q,q')x + b(q,q')y + c(q,q')$.
	Let $A(q,q')=a(q,q')\cdot\rho(q,q')$, $B(q,q')=b(q,q')\cdot\rho(q,q')$,
	and $C(q,q')=c(q,q')\cdot\rho(q,q')$.  (As noted earlier, we can
	precompute
	all the $\rho(q,q')$ values in $O(n^2)$ time from the
	dual arrangement of lines.)
	By \eqref{eqn:ch},
	\[
		\phi(p,\V{ch}) \ = \sum_{\mbox{\scriptsize\begin{tabular}{c}$q,q'\in P\ (q\neq q')$\\ with $p\in H(q,q')$\end{tabular}}}\!\! 
			(A(q,q')x+B(q,q')y+C(q,q'))\ \ =\ \ {\cal A}(p)x+{\cal B}(p)y+{\cal C}(p),
	\]
	where
	\[ 
	{\cal A}(p)\ = \!\!\!\sum_{\mbox{\scriptsize\begin{tabular}{c}$q,q'\in P\ (q\neq q')$\\ with $p\in H(q,q')$\end{tabular}}}\!\!\!\!\! A(q,q'),\ \ \
	{\cal B}(p)\ = \!\!\!\sum_{\mbox{\scriptsize\begin{tabular}{c}$q,q'\in P\ (q\neq q')$\\ with $p\in H(q,q')$\end{tabular}}}\!\!\!\!\! B(q,q'),\ \ \
	{\cal C}(p)\ = \!\!\!\sum_{\mbox{\scriptsize\begin{tabular}{c}$q,q'\in P\ (q\neq q')$\\ with $p\in H(q,q')$\end{tabular}}}\!\!\!\!\! C(q,q').
	\]
	We describe how to compute ${\cal A}(p)$, ${\cal B}(p)$, and ${\cal C}(p)$
	for all $p\in P$ in $O(n^2)$ total time.  Afterwards, we can compute $\phi(p,\V{ch})$ for all $p\in P$ in $O(n)$ additional time.

	The problem can be reduced to 3 instances of the following:
	\begin{quote}
		Given a set $P$ of $n$ points in the plane, and given $O(n^2)$ 
		lines each through 2 points of $P$ and each assigned a weight, 
		compute for all $p\in P$ the sum of the weights of all lines 
		below $p$ (or similarly all lines above $p$).
	\end{quote}
	In the dual, the problem becomes:
	\begin{quote}
		Given a set $L$ of $n$ lines in the plane, 
		and given $O(n^2)$ vertices in the arrangement, each assigned a weight, 
		compute for all lines $\ell\in L$ the sum $S(\ell)$ of the weights
		of all vertices below $\ell$.
	\end{quote}

	To solve this problem, we could use known data structures for 
	halfplane range searching, but a direct solution is simpler.
	First construct the arrangement in $O(n^2)$ time.
	Given $\ell,\ell'\in L$,
	define $S(\ell,\ell')$ to be the sum of the
	weights of all vertices on the line $\ell'$ that are below $\ell$.  
	For a fixed $\ell'\in L$, 
	we can precompute $S(\ell,\ell')$ for all $\ell\in L\setminus\{\ell'\}$ 
	in $O(n)$ time, since these values correspond to prefix or suffix sums over
	the sequence of weights of the $O(n)$ vertices on the line $\ell'$.  
	The total time for all lines $\ell'\in L$ is $O(n^2)$.

	Afterwards, for each $\ell\in L$, we can compute $S(\ell)$ in $O(n)$ time by summing
	$S(\ell,\ell')$ over all $\ell'\in L\setminus\{\ell\}$ and dividing by 2 (since each	
	vertex is counted twice).
	The total time for all $\ell\in L$ is $O(n^2)$.
\end{proof}

\begin{theorem}
\label{thm:pch}
	The Shapley values of the \PerimeterConvexHull game for $n$ points can be computed 
	in $O(n^2)$ time.
\end{theorem}
\begin{proof}
	We adapt the algorithm for area to handle the perimeter.
	Let $E(\pi,p)$ be the set of directed edges of $\CH(P(\pi,p))$ 
	that are not in $\CH(P(\pi,p)\setminus\{p\})$, where edges
	are oriented in clockwise order.  (If $\CH(P(\pi,p)= \CH(P(\pi,p)\setminus\{p\})$, 
	then $E(\pi,p)$ is empty; otherwise, it has exactly	two edges because of general position.)  
	Then $(q,p)\in E(\pi,p)$ if and only if $q$ appears before $p$ in $\pi$
	and no points before $p$ in $\pi$ lie in $H(q,p)$; in other words,
	among the $\level(q,p)+2$ points in $H(q,p)\cup\{p,q\}$,
	the earliest point is $q$ and the second earliest
	point is $p$.  We can use Lemma~\ref{le:permutations2} to bound the probability of this 
	event. Thus, for fixed $p,q\in P$, the probability that $(q,p)\in E(\pi,p)$ 
	(with respect to a random permutation $\pi$) is exactly
	\[ 
		\rho'(q,p) ~=~ \frac{1}{(\level(q,p)+2)(\level(q,p)+1)}.
	\]  
	Similarly, the probability that $(p,q)\in E(\pi,p)$ is exactly $\rho'(p,q)$.
	It follows that the expected total length of the edges in $E(\pi,p)$ is
	\[ 
		\phi^+(p)\ = \sum_{q\in P\setminus\{p\}} \|p-q\|\cdot (\rho'(q,p)+\rho'(p,q)).
	\]
 
	On the other hand, a modification of the proof of
	Lemma~\ref{le:convex_hull} shows that
	the expected total length of the edges in $\CH(P(\pi,p)\setminus\{p\})$
	that are not in $\CH(P(\pi,p))$ is
	\[
		\phi^-(p)\ = \sum_{\mbox{\scriptsize\begin{tabular}{c}$q,q'\in P\ (q\neq q')$\\ with $p\in H(q,q')$\end{tabular}}}\!\! 
			\|q-q'\|\cdot\rho(q,q').
	\]
  
	The Shapley value of each point $p\in P$ is $\phi(p,\V{pch})=\phi^+(p)-\phi^-(p)$ because of linearity of expectation.
	We can compute $\phi^+(p)$ naively in $O(n)$ time for each $p\in P$;
	the total time is $O(n^2)$.
	We can compute $\phi^-(p)$ for all $p\in P$ as in the proof in Theorem~\ref{thm:ach},
	in $O(n^2)$ total time (in fact, the algorithm is a little simpler, since linear functions are not required).
\end{proof}

\section{Smallest enclosing disk}
\label{sec:enclosingdisk}

In this section we consider the area and the perimeter
of the smallest enclosing disk of the points. Recall that
$\V{ed}$ is the characteristic function.
For simplicity, we assume general position in the following way:
no four points are cocircular and circles through three input points 
do not have a diameter defined by two input points.
We use $P$ for the set of points.

First we explain how to compute the Shapley values for the area of
the minimum enclosing disk.

\begin{lemma}
\label{le:encloding_disk}
	For each point $p$ of $P$ we can compute $\phi(p,\V{ed})$ in
	$O(n^3)$ time. 
\end{lemma}
\begin{proof}
	For each subset of points $Q\subset P$, 
	let $X(Q)$ be the subset of points of $Q$ that lie on the boundary of $\med(Q)$. 
	We now recall some well known properties; see for example \cite{Welzl91} 
	or \cite[Section 4.7]{BergCKO08}.
	The disk $\med(Q)$ is unique and $\med(X(Q))=\med(Q)$.
	If a point $p$ is outside $\med(Q)$, then $p$ is on the boundary of $\med(Q\cup \{ p\})$,
	that is, $p\in X(Q\cup \{p \})$.
	Because of our assumption on general position, $X(Q)$ has at most three points.
	When $X(Q)$ has two points, then they define a diameter of $\med(Q)$.
	We have $|X(Q)|\le 1$ only when $|Q|\le 1$.
	See Figure~\ref{fig:disk1}.

	\begin{figure}
		\centering
		\includegraphics[width=.55\textwidth,page=1]{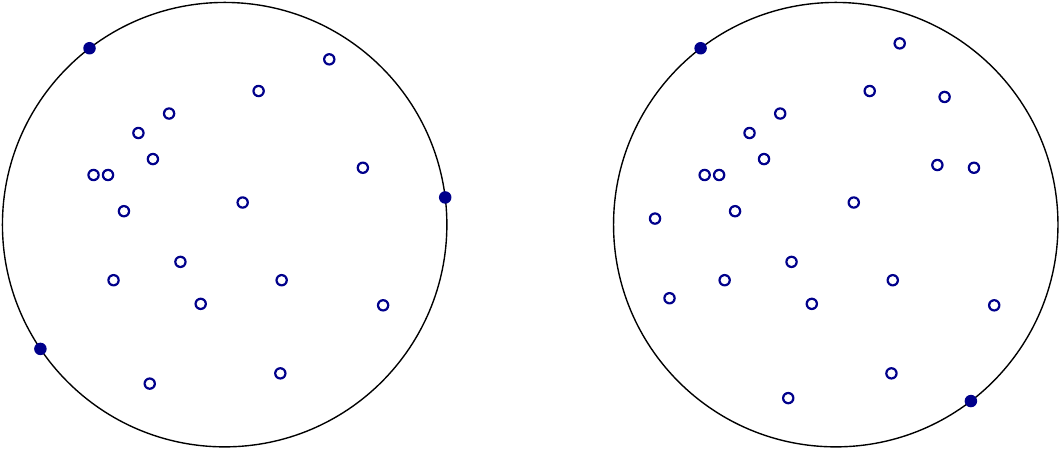}
		\caption{The minimum enclosing disk for two sets $Q$ of points. 
			Left: $X(Q)$ has three elements. Right: $X(Q)$ has two elements.}
		\label{fig:disk1}
	\end{figure}

	A subset $B\subset P$ of size at most 3 such that $X(B)=B$ is called a \emph{basis}.
	For a basis $B$, define $\Out(B)= P\setminus \med(B)$
	and $\level(B)=|\Out(B)|$. Note that $\level(B)$ is the number of points of $P$
	strictly outside $\med(B)$.
  	
	For a basis $B$ and a point $p\not\in\med(B)$,
	we have $X(P(\pi,p)\setminus\{p\})=B$ if and only if all points of $B$ appear
	before $p$ in $\pi$, and no points before $p$ in $\pi$ belong to $\Out(B)$.  
	In other words, among the $\level(B)+|B|$ points in
	$\Out(B))\cup B$, the $|B|$ earliest points are the points of $B$
	and the next earliest point is $p$.  See Figure~\ref{fig:disk2} (left).
	Lemma~\ref{le:permutations2} implies that, for a fixed $B$ and a fixed $p$, 
	the probability that $X(P(\pi,p)\setminus\{p\})=B$ (with respect to a random permutation $\pi$)
	is exactly
	\[ 
		\rho(B) ~=~ \frac{(\level(B)-1)! \,|B|!}{(\level(B)+|B|)!}
				~=~ \frac{|B|!}{(\level(B)+|B|)\cdots(\level(B)+1)\level(B)}.
	\]
	
	Let $I(\pi,p)$ be the indicator variable that is 
	$1$ if $p\not\in\med(P(\pi,p)\setminus\{p\})$, and $0$ otherwise.  
	It follows that the expected value of
	$\area(\med(P(\pi,p)\setminus\{p\}))\cdot I(\pi,p)$ is 
	\begin{equation}\label{eqn:med1}
	\phi^-(p)\ = \sum_{\mbox{\scriptsize \begin{tabular}{c} basis $B$\\ with $p\not\in\med(B)$\end{tabular}}} \!\!\area(\med(B))\cdot\rho(B).
	\end{equation}

	\begin{figure}
	\centering
	\includegraphics[width=.9\textwidth,page=3]{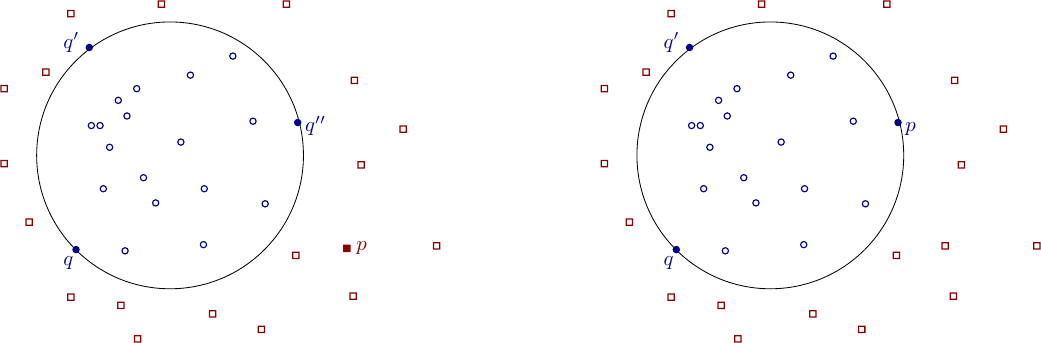}
	\caption{Left: in order for the shown disk to be $\med(P(\pi,p)\setminus\{p\})$, 
		the points $q,q',q''$ must appear before $p$, which
		in turn must appear before all other points outside the disk. 
		Right: in order for the shown disk to be $\med(P(\pi,p))$, 
		the points $q,q'$ must appear before $p$, which in turn must appear 
		before all points outside the disk.}
	\label{fig:disk2}
	\end{figure}
	
	On the other hand, for a basis $B$ and a point $p\in B$,
	we have $X(P(\pi,p))=B$ if and only if all points of $B\setminus\{p\}$
	appear before $p$ in $\pi$, and no points before $p$ in $\pi$ lie in $\Out(B)$.
	In other words, among the $\level(B)+|B|$ points in
	$\Out(B)\cup B$, the $|B|-1$ earliest points are the points of $B\setminus\{p\}$
	and the next earliest point is $p$. See Figure~\ref{fig:disk2} (right).
	We can bound the probability of this event using Lemma~\ref{le:permutations2}.
	Thus, for a fixed $B$ and $p\in B$, 
	the probability that $X(P(\pi,p)\setminus\{p\})=B$ 
	(with respect to a random permutation $\pi$)
	is exactly
	\[ \rho'(B) ~=~ \frac{(|B|-1)!}{(\level(B)+|B|)\cdots(\level(B)+1)}.\]
	
	It follows that the expected value of
	$\area(\med(P(\pi,p))\cdot I(\pi,p)$ is 
	\begin{equation}\label{eqn:med2}
	\phi^+(p)\ = \sum_{\mbox{\scriptsize \begin{tabular}{c} basis $B$\\ with $p\in B$\end{tabular}}} \!\!\area(\med(B))\cdot\rho'(B).
	\end{equation}
	
	By linearity of expectation, the Shapley value of $p$ is $\phi(\V{ed},p)=\phi^+(p)-\phi^-(p)$.  
	From the formulas (\ref{eqn:med1}) and (\ref{eqn:med2}),
	we can compute $\phi^+(p)$ in
	$O(n^3)$ time and $\phi^-(p)$ 
	in $O(n^2)$ time for any given $p\in P$ (since $|B|\le 3$), 
	if all $\rho(B)$ and $\rho'(B)$ values have been
	precomputed.
	
	Precomputing $\rho(B)$ and $\rho'(B)$ requires precomputing
	$\level(B)$ for all bases $B$.  Precomputing $\level(B)$
	for bases $B$ of size~2 can be naively done in $O(n^2\cdot n)=
	O(n^3)$ time, so it suffices to focus on bases $B$ of size~3.
	By a standard lifting transformation (mapping point $(a,b)$ to the plane
	$z-2ax-2by+a^2+b^2=0$ in 3 dimensions), the problem reduces to 
	compute the \emph{level} of $O(n^3)$ vertices in an arrangement
	of $n$ planes in 3 dimensions. See~\cite[Section 5.7]{Matousek02} 
	for a discussion on the transformation.
	The arrangement of $n$ planes can be constructed in $O(n^3)$ time~\cite{EdelsbrunnerOS86,EdelsbrunnerSS93}, 
	and the levels of all vertices can be
	subsequently generated by traversing the arrangement in $O(1)$ time per vertex.
	
	For the perimeter, we just use the perimeter of $\med(B)$ instead of the area
	in the computations.
\end{proof}

Naively applying Lemma~\ref{le:encloding_disk} to all points $p\in P$ 
gives $O(n^4)$ total time. We can speed up the algorithm by a factor of $n$:

\begin{theorem}
\label{thm:ed}
	The Shapley values of the \AreaEnclosingDisk and \PerimeterEnclosingDisk games 
	for $n$ points can be computed in $O(n^3)$ time. 	
\end{theorem}
\begin{proof}
	We already know how to compute $\phi^+(p)$ in
	$O(n^2)$ time for each $p\in P$, and thus in $O(n^3)$ total time.
	It suffices to focus on computing $\phi^-(p)$.
	In the formula (\ref{eqn:med1}), terms for bases of size~2
	can be handled in $O(n^2)$ time for each $p$; so it suffices to
	focus on bases of size~3. The problem can be formulated as follows:
	\begin{quote}
	Given a set $P$ of $n$ points in the plane, and given $O(n^3)$ 
	disks each with 3 points of $P$ on the boundary and each assigned a weight, 
	compute for all $p\in P$ the sum of the weights of all disks not containing $p$.
	\end{quote}
	By the standard lifting transformation, the problem reduces to the following:
	\begin{quote}
	Given a set $H$ of $n$ planes in 3 dimensions, and given $O(n^3)$ vertices in the arrangement, 
	each assigned a weight, compute for all $h\in H$ the sum $S(h)$ of the weights
	of all vertices below $h$.
	\end{quote}

	To solve this problem, we could use known data structures for 
	halfspace range searching, but an approach as in the proof
	of Theorem~\ref{thm:ach} is simpler.
	First construct the arrangement in $O(n^3)$ time.
	Given $h,h',h''\in H$, define $S(h,h',h'')$  to be the sum of the
	weights of all vertices on the line $h'\cap h''$ that are below $h$.
	For a fixed pair of planes $h',h''\in H$,
	we can precompute $S(h,h',h'')$ for all $h\in H\setminus\{h',h''\}$ in $O(n)$ time, since these values correspond to prefix or suffix sums over
	the sequence of weights of the $O(n)$ vertices on the line $h'\cap h''$.  The total time for all pairs $h',h''\in H$ is $O(n^3)$.

	Afterwards, for each $h\in H$, we can compute $S(h)$ in $O(n^2)$ time by summing
	$S(h,h',h'')$  over all pairs $h',h''\in H\setminus\{h\}$ 
	and dividing by 3 (since each vertex is counted thrice).
	The total time for all $h\in H$ is $O(n^3)$.
\end{proof}

\section{Union of anchored rectangles}
\label{sec:anchoredrectangles}

In this section we consider the \AreaAnchoredRectangles game defined
by the characteristic function $\V{ar}$.
It is easy to see that one can focus on the special case where
all the points are in a quadrant; see Figure~\ref{fig:reduction1}.
Our discussion will focus on the case where the points of $P$ 
are on the positive quadrant of the plane.

\begin{figure}
	\centering
	\includegraphics[width=\textwidth,page=2]{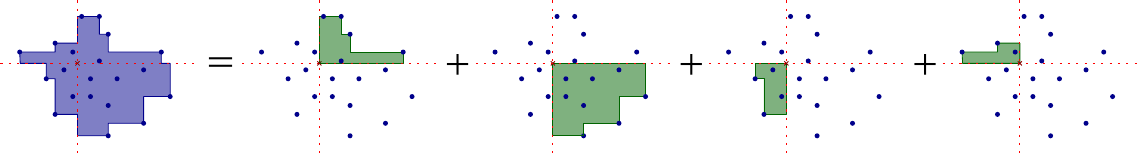}
	\caption{It is enough to consider the cases where all points are in a quadrant.}
	\label{fig:reduction1}
\end{figure}

Consider a fixed set $P$ of points in the positive quadrant.
In the notation we will drop the dependency on $P$. 
For simplicity, we assume general position: no two points have the same $x$- or $y$-coordinate.
We first introduce some notation that will be used in this section and
in Section~\ref{sec:boundingbox}.

\subsection{Notation for axis-parallel problems}
\label{sec:notation}

For each point $q$ of the plane, we use the ``cardinal directions"
to define subsets of points in quadrants with apex at $q$:
\begin{align*}
	\SNW(q) ~&=~ \{ p\in P \mid x(p)\le x(q), ~ y(p)\ge y(q)\},\\
	\SNE(q) ~&=~ \{ p\in P \mid x(p)\ge x(q), ~ y(p)\ge y(q)\},\\
	\SSE(q) ~&=~ \{ p\in P \mid x(p)\ge x(q), ~ y(p)\le y(q)\}.
\end{align*}
We use lowercase to denote their cardinality: $\NW(q)=|\SNW(q)|$,
$\NE(q)=|\SNE(q)|$ and $\SE(q)=|\SSE(q)|$.
See Figure~\ref{fig:notation}, left.

\begin{figure}
	\centering
	\includegraphics[]{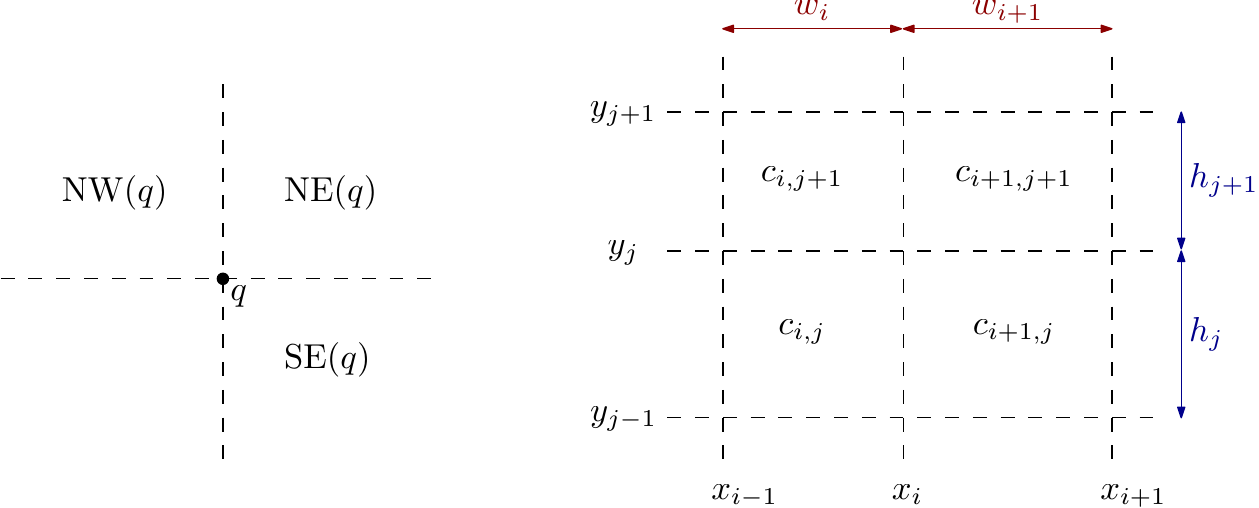}
	\caption{Notation for axis-parallel problems. Left: the quadrants to define $\SNW(q)$, $\SNE(q)$ and $\SSE(q)$.
		Right: cells of $\A$.}
	\label{fig:notation}
\end{figure}

Let $x_1< \ldots < x_n$ denote the $x$-coordinates of the points of $P$,
and let $y_1< \ldots < y_n$ be their $y$-coordinates.
We also set $x_0=0$ and $y_0=0$.
For each $i,j\in [n]$ we use $w_i=x_i-x_{i-1}$ (for \emph{width})
and $h_j=y_{j}-y_{j-1}$ (for \emph{height}).

Let $L$ be the set of horizontal and vertical lines that contain
some point of $P$. We add to $L$ both axis. Thus, $L$ has $2n+2$ lines.
The lines in $L$ break the plane into 2-dimensional cells (rectangles), 
usually called the arrangement and denoted by $\A=\A(L)$. 
More precisely, a ($2$-dimensional) cell $c$ of $\A$ is a maximal connected 
component in the plane after the removal of the points on lines in $L$. 
Formally, the ($2$-dimensional) cells are open sets whose closure is a rectangle,
possibly unbounded in some direction.
We are only interested in the bounded cells, and with a slight abuse
of notation, we use $\A$ for the set of bounded cells.
We denote by $c_{i,j}$ the cell between the vertical lines $x=x_{i-1}$ and $x=x_i$
and the horizontal lines $y=y_{j-1}$ and $y=y_j$.
Note that $c_{i,j}$ is the interior of a rectangle with width $w_i$ and height $h_j$.
See Figure~\ref{fig:notation}, right.

Since $\SNE(q)$ is constant over each $2$-dimensional cell $c$ of $\A$,
we can define $\SNE(c)$, for each cell $c\in \A$. 
The same holds for $\SNW(c)$ and $\SSE(c)$ and their cardinalities,
$\NE(c)$, $\NW(c)$ and $\SE(q)$.
See Figure~\ref{fig:ortho_hull1}.

\begin{figure}
	\centering
	\includegraphics[scale=.9,page=1]{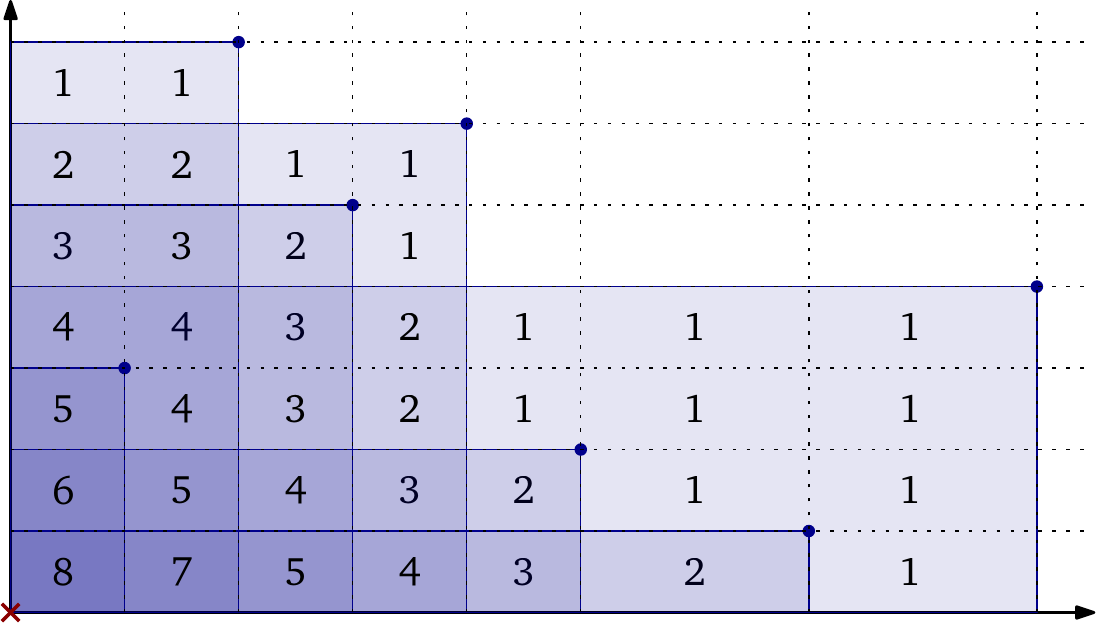}
	\caption{The non-zero counters $\NE(c)$ for the bounded cells $c$ of $\A$. 
		The intensity of the color correlates with the counter $\NE(c)$.}
	\label{fig:ortho_hull1}
\end{figure}

A \DEF{block} is a set of cells
$B=B(i_0,i_1, j_0, j_1) = \{c_{i,j}\mid i_0\le i\le i_1,~ j_0\le j\le j_1\}$ 
for some indices $i_0,i_1, j_0, j_1$,
with $1\le i_0\le i_1\le n$ and $1\le j_0\le j_1\le n$.
The number of columns and rows in $B$ is $i_1-i_0+1+ j_1-j_0+1= O(i_1-i_0+j_1-j_0)$.
A block $B$ is \DEF{empty} if no point of $P$
is on the boundary of at least three cells of $B$.
Equivalently, $B$ is empty if no point of $P$ is in the interior 
of the union of the closure of the cells in $B$.
See Figure~\ref{fig:notation2} for an example.

\begin{figure}
	\centering
	\includegraphics[scale=1,page=2]{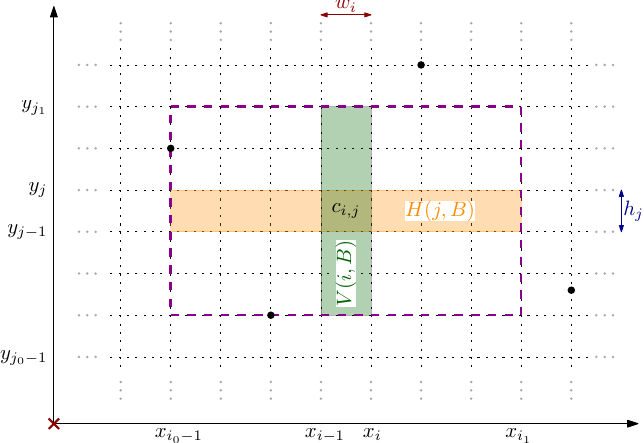}
	\caption{An \emph{empty} block $B=B(i_0,i_1,j_0,j_1)$ 
		with a vertical and a horizontal slab shaded.}
	\label{fig:notation2}
\end{figure}

We will be using maximal rows and columns within a block $B$
to compute some partial information.
Thus, for each block $B$ and each index $i$, we define the \DEF{vertical slab}
$V(i,B)=\{ c_{i,j}\mid 1\le j\le n, ~ c_{i,j}\in B\}$.
Similarly, for each block $B$ and each index $j$, we define the \DEF{horizontal slab}
$H(j,B)=\{ c_{i,j}\mid 1\le i\le n, ~ c_{i,j}\in B\}$.
Such slabs are meaningful only for indices within the range that defines the slab.
We call them the \DEF{slabs within} $B$.

\subsection{Interpreting Shapley values geometrically}

First, we reduce the problem of computing Shapley values 
to a purely geometric problem.
See Figure~\ref{fig:ortho_hull1} for the relevant counters considered in the following result.

\begin{lemma}
\label{le:percell}
	If $P$ is in the positive quadrant, then for each $p\in P$ we have
	\[
		\phi(p,\V{ar}) ~=~ \sum _{c\in \A,~ c\subset R_p} \frac{\area(c)}{\NE(c)} \,.
	\]
\end{lemma}
\begin{proof}
	Each cell $c$ of the arrangement $\A$ defines a game for the set of players $P$ 
	with characteristic function
	\[
		v_c(Q) ~=~\begin{cases}
			\area(c) & \text{if $c\subset R_p$ for some $p\in Q$,}\\
			0 & \text{otherwise.}
			\end{cases}
	\]
	To analyze the Shapley values of the game defined by $v_c$,
	note that, because of symmetry, each point of the set 
	$\SNE(c) = \{p\in P\mid c\subset R_p \}$
	has the same probability of being the first point from $\SNE(c)$ in 
	a random permutation. Therefore 
	\[
		\phi(p,v_c) ~=~ \begin{cases}
					\frac{\area(c)}{\NE(c)}&\text{if $c\subset R_p$,}\\
					0 &\text{otherwise.}
					\end{cases}
	\]
	On the other hand, 
	\[
		\V{ar}(Q) ~=~ \sum_{c\in \A} v_c(Q)~~\text{ for all $Q\subset P$,}
	\]
	and because of linearity of Shapley values we get
	\[
		\phi(p,\V{ar}) ~=~ \sum_{c\in \A} \phi(p,v_c) 
						~=~ \sum_{c\in \A,~c\subset R_p} \phi(p,v_c) 
						~=~ \sum_{c\in \A,~c\subset R_p} \frac{\area(c)}{\NE(c)}.
		\qedhere
	\]
\end{proof}

For each subset $C$ of cells of $\A$, we define 
\[
	\sigma(C) ~=~ \sum_{c\in C} \frac{\area(c)}{\NE(c)} 
\]
(We will only consider sets $C$ of cells with $\NE(c)>0$ for all $c\in C$.)
Note that we want to compute $\sigma(\cdot)$ for the sets of cells contained
in the rectangles $R_{p}$ for all $p\in P$. 

Using standard tools in computational geometry we can compute the values 
$\phi(p,\V{ar})$ for all $p\in P$ in near-quadratic time. 
Here is an overview of the approach.
An explicit computation of $\A$ takes quadratic time in the worst case.
Using standard data structures for orthogonal range searching, see~\cite{Willard85}
or \cite[Chapter 5]{BergCKO08},
we can then compute $\NE(c)$ for each cell $c\in \A$.
Finally, replacing each cell $c$ by a point $q_c\in c$ with weight $w_c=\area(c)/\NE(c)$,
we can reduce the problem of computing $\phi(p,\V{ar})$ to the problem
of computing $\sum_{q_c\in R_p} w_c$, which is again an orthogonal range query.
An alternative is to use dynamic programming across the cells of $\A$ to compute
first $\NE(c)$ and then partial sums of the weights $w_c$. 

Our objective in the following sections is to improve this result using the
correlation between adjacent cells. It could seem at first glance that
segment trees~\cite[Section 10.3]{BergCKO08} may be useful. 
We did not see how to work out the details of this; 
the problem is that the weights are inversely proportional to $\NE(c)$.

\subsection{Handling empty blocks}
\label{sec:blocks}

In the following we assume that we have preprocessed $P$ in $O(n\log n)$ time
such that $\NE(q)$ can be computed in $O(\log n)$ time
for each point $q$ given at query time~\cite{Willard85}.
This is a standard range counting for orthogonal ranges and
the preprocessing has to be done only once.

When a block $B$ is empty, then we can use multipoint evaluation
to obtain the partial sums $\sigma(\cdot)$ for each vertical and
horizontal slab of the block.

\begin{lemma}
\label{le:block}
	Let $B$ be an empty block with $k$ columns and rows.
	We can compute in $O(k \log n )$ time the 
	values $\sigma(C)$ for all slabs $C$ within $B$.
\end{lemma}
\begin{proof}
	Assume that $B$ is the block $B(i_0,i_1,j_0,j_1)$.
	We only explain how to compute the values $\sigma(V(i,B))$ for all $i_0\le i\le i_1$.
	The computation for the horizontal slabs $\sigma(H(j_0,B)),\dots,\sigma(H(j_1,B))$ is similar.

	We look into the first vertical slab $V(i_0,B)$ and make groups of cells depending on their
	value $\NE(\cdot)$.
	More precisely, for each $\ell$ we define
	$J(\ell)=\{ j\mid j_0\le j \le j_1, ~ \NE(c_{i_0,j})=\ell\}$.
	Let $\ell_0$ and $\ell_1$ be the minimum and the maximum $\ell$ such that $J(\ell)\neq 0$, respectively.	
	
	We set up the following rational function with variable $x$:
	\[
		R(x) ~=~ \sum_{\ell=\ell_0}^{\ell_1} \frac{~\sum_{j\in J(\ell)}h_j~}{\ell+x}.
	\]
	Setting $t= \ell-\ell_0$, $b_t=\sum_{j\in J(\ell_0+t)}h_j$ and $\Delta=\ell_0$, we have
	\[
		R(x) ~=~ \sum_{t=0}^{\ell_1-\ell_0} \frac{b_t}{\Delta+t+x}.
	\]	
	Thus, this is a rational function of the shape considered in Lemma~\ref{le:multipoint_evaluation}
	with $\ell_1-\ell_0\le j_1-j_0+1\le k$ terms.
	The coefficients can be computed in $O(k\log n)$ time
	because we only need the values $h_j$ and $\NE(c_{i_0,j})$ for each $j$.
	These latter values $\NE(\cdot)$ are obtained from range counting queries.
	
	Note that 
	\[
		w_{i_0}\cdot R(0) ~=~ w_{i_0}\cdot ~ \sum_{\ell=\ell_0}^{\ell_1} ~\sum_{j\in J(\ell)} \frac{h_j}{\ell}
		~=~ \sum_{j=j_0}^{j_1} \frac{w_{i_0}h_j}{\NE(c_{i_0,j})} 
		~=~ \sum_{j=j_0}^{j_1} \frac{\area(c_{i_0,j})}{\NE(c_{i_0,j})} ~=~
		\sigma(V(i_0,B)).
	\]
	A similar statement holds for all the other vertical slabs within $B$.
	We make the statement precise in the following.

	Consider two consecutive vertical slabs $V(i,B)$ and $V(i+1,B)$ within  
	the block $B$. 
	Because the block $B$ is empty,
	the difference $\NE(c_{i+1,j})-\NE(c_{i,j})$ is independent of $j$.
	See Figure~\ref{fig:block}.
	It follows that, for each index $i$ with $i_0\le i\le i_1$,
	there is an integer $\delta_i$ such that $\NE(c_{i,j})=\NE(c_{i_0,j})+\delta_i$ 
	for all $j$ with $j_0\le j\le j_1$. 
	Moreover, for each $i$ with $i_0\le i\le i_1$ and each $\ell$ with $\ell_0\le \ell\le \ell_1$,
	the value of $\NE(c_{i,j})$ is constant over all $j\in J(\ell)$.
	Therefore, for each $j\in J(\ell)$ we have $\NE(c_{i,j})= \ell+\delta_i$.
	
	Each value $\delta_i$ can 
	be obtained using that $\delta_i = \NE(c_{i,j_0})-\NE(c_{i_0,j_0})$.
	This means that the values $\delta_{i_0},\dots,\delta_{i_1}$
	can be obtained in $O(k \log n)$ time.
		
	\begin{figure}
		\centering
		\includegraphics[scale=1,page=9]{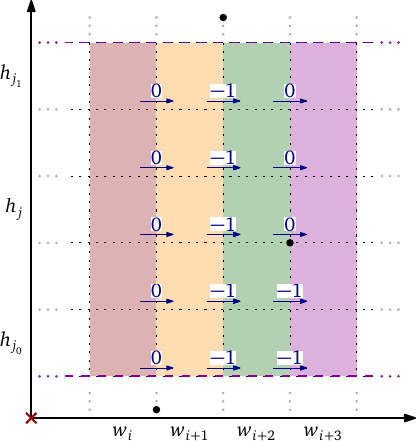}
		\caption{Changes in the values of $\NE(\cdot)$ when passing from 
			a vertical slab to the next one. The rightmost transition
			shows the need to deal with empty blocks for our argument.}
		\label{fig:block}
	\end{figure}

	Now we note that, for each index $i$ with $i_0\le i \le i_1$, we have
	\[
		w_{i}\cdot R(\delta_i) ~=~ w_{i}\cdot ~ \sum_{\ell=\ell_0}^{\ell_1} ~\sum_{j\in J(\ell)} \frac{h_j}{\ell+\delta_i}
		~=~ \sum_{j=j_0}^{j_1} \frac{w_i h_j}{\NE(c_{i,j})} 
		~=~ \sum_{j=j_0}^{j_1} \frac{\area(c_{i,j})}{\NE(c_{i,j})} ~=~
		\sigma(V(i,B)).
	\]
	We use Lemma~\ref{le:multipoint_evaluation} to evaluate 
	the $i_1-i_0+1\le k$ values $R(\delta_i)$, where $i_0\le i \le i_1$, 
	in $O(k\log k)=O(k\log n)$ time. 
	After this, we get each value $\sigma(V(i,B))=w_i\cdot R(\delta_i)$ in constant time.
\end{proof}

\subsection{Chains}
\label{sec:ar_chain}
In this section we consider the case where the points are a chain.
As discussed before, it is enough to consider that $P$ is in the positive
quadrant.
After sorting, we can assume without loss of generality
that the points of $P$ are indexed so that $0<x(p_1)<\dots < x(p_n)$.

We start with the easier case: increasing chains.
The problem is actually an \Airport game in disguise.

\begin{figure}
	\centering
	\includegraphics[page=7]{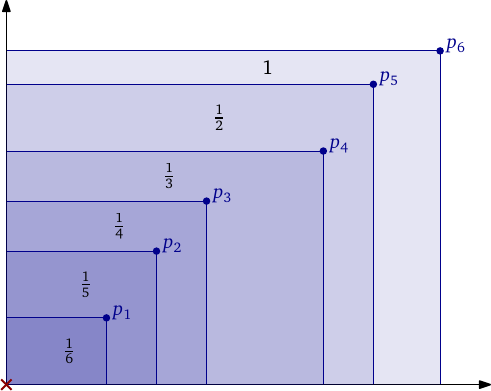}
	\caption{Increasing chain in the positive quadrant. 
		Each region is marked with the multiplicative weight for its area.}
	\label{fig:increasing}
\end{figure}

\begin{lemma}
\label{le:ar_chain1}
	If $P$ is an increasing chain in the positive quadrant, 
	then we can compute the Shapley values of the \AreaAnchoredRectangles 
	game in $O(n\log n)$. 
\end{lemma}
\begin{proof}
	Set the point $p_0$ to be the origin $o$.
	Note that, for each $i\in [n]$ and each cell $c\in \A$ contained 
	$R_{p_i}\setminus R_{p_{i-1}}$ we have $\NE(c)= n-i+1$.
	For each $i\in [n]$, define the value 
	\[
		z_i ~=~ \frac{\area(R_{p_i})-\area(R_{p_{i-1}})}{n-i+1}.
	\]
	Thus, $z_i$ is the area of the region $R_{p_i}\setminus R_{p_{i-1}}$ divided
	by $\NE(c)$, for some $c\subset R_{p_i}\setminus R_{p_{i-1}}$.
	See Figure~\ref{fig:increasing}.
	Because of Lemma~\ref{le:percell}, we have for
	each $i\in [n]$
	\[
		\phi(p_i,\V{ar}) ~=~ \sum_{c\in \A, ~ c\subset R_{p_i}} \frac{\area(c)}{\NE(c)} 
			~=~ \sum_{j\le i} z_j.
	\]
	Therefore $\phi(p_1,\V{ar})=z_1$ and, for each $i\in [n]\setminus \{1\}$,
	we have $\phi(p_i,\V{ar})=z_i+\phi(p_{i-1},\V{ar})$.
	The result follows.
	
	(Actually this game is just the $1$-dimensional \Airport game if each point $p_i$ 
	is represented on the real line with the point with $x$-coordinate $\area(R_{p_i})$.)
\end{proof}

It remains the more interesting case, when the chain is decreasing.
See Figure~\ref{fig:decreasing1} for an example.
It is straightforward to see that, if $i+j> n+1$, then $\NE(c_{i,j})=0$,
and if $i+j\le n+1$, we have $\NE(c_{i,j})=n+2-i-j$.
So in this case we do not really need data structures to obtain $\NE(c_{i,j})$ efficiently.
(The proof of Lemma~\ref{le:block} can be slightly simplified for this case because $J(\ell)$
contains a single element due to the special structure of the values $\NE(c)$.
See Figure~\ref{fig:decreasing2}.)

\begin{figure}
	\centering
	\includegraphics[width=\textwidth,page=2]{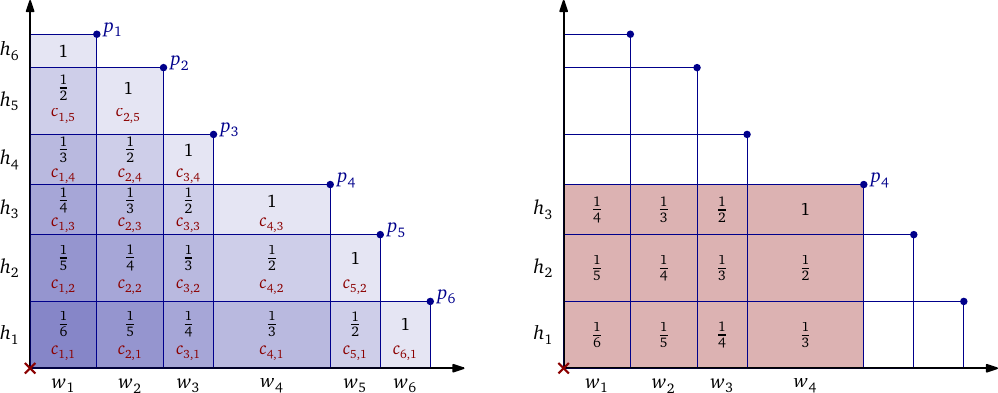}
	\caption{Decreasing chain. Left: Each cell $c_{i,j}$ is marked with the multiplicative weight 
		$\tfrac{1}{\NE(c)}$ for its area.
		Right: the cells whose contribution we have to add for the point $p_4$.}
	\label{fig:decreasing1}
\end{figure}

\begin{figure}
	\centering
	\includegraphics[page=6]{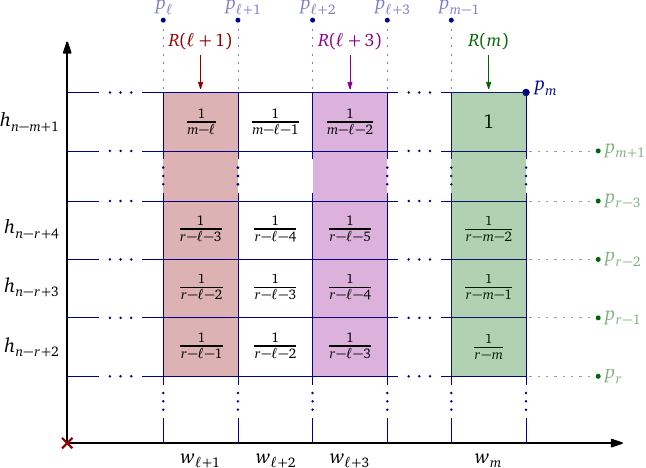}
	\caption{Values $1/\NE(c)$ in a decreasing chain.}
	\label{fig:decreasing2}
\end{figure}

We use a divide-and-conquer paradigm considering certain empty blocks
defined by two indices $\ell$ and $r$, where $\ell< r$.
Since the indexing of rows is not the most convenient in this case,
it is better to introduce the notation $B_{\ell,r}$ 
for the block $B(\ell+1,m,n-r+2,n-m+1)$, 
where $m=m(\ell,r)=\lfloor(\ell+r)/2\rfloor$.
Initially we will have $\ell=0$ and $r=n+1$, which means
that we start with the block $B_{0,n+1}=B(1,m,1,m)$. 
See Figure~\ref{fig:base_case} for the base case,
and Figure~\ref{fig:general_case} for a generic case.
For each block $B_{\ell,r}$,
we will compute $\sigma(C)$ for each slab $C$ within $B_{\ell,r}$.
The blocks can be used to split the problem into two smaller
subchains, and the interaction between them in encoded by the
slabs within the block.
This approach leads to the following result.

\begin{figure}
	\centering
	\includegraphics[page=3,scale=.95]{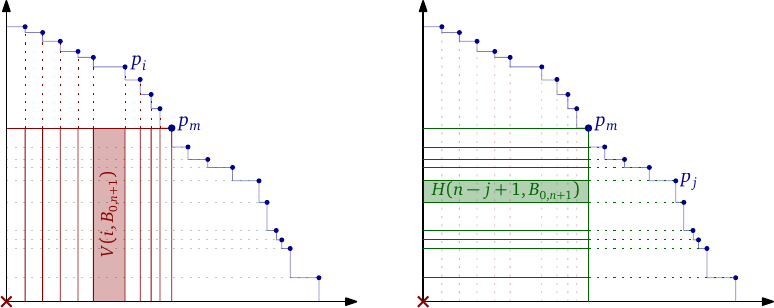}
	\caption{Vertical and horizontal slabs for the base case,
		where we consider the block $B_{0,n+1}$.}
	\label{fig:base_case}
\end{figure}
\begin{figure}
	\centering
	\includegraphics[page=4,scale=.95]{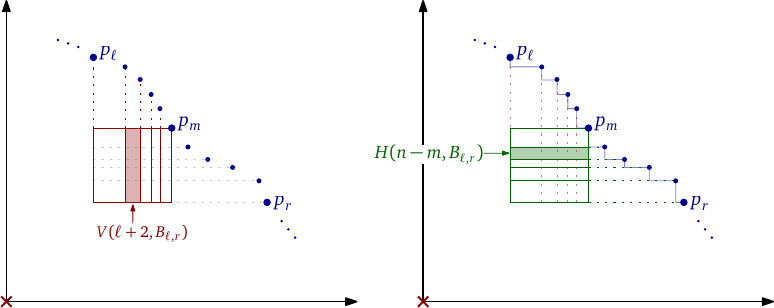}
	\caption{Vertical and horizontal slabs for the divide-and-conquer (general case).}
	\label{fig:general_case}
\end{figure}

\begin{lemma}
\label{le:ar_chain2}
	If $P$ is a decreasing chain with $n$ points in the positive quadrant, 
	then we can compute the Shapley values of 
	the \AreaAnchoredRectangles game in $O(n\log^2 n)$ time.	
\end{lemma}
\begin{proof}
	We assume for simplicity that $n+1$ is a power of $2$. Otherwise
	we replace each appearance of an index larger than $n+1$ by $n+1$.
	
	Let $\I$ be the set of pairs $(\ell,r)$ 
	defined by $\ell=\alpha 2^\beta$ and $r=(\alpha+1)2^\beta \le n+1$,
	where $\alpha$ and $\beta$ are non-negative integers. 
	For each $(\ell,r)\in \I$, we compute the values 
	$\sigma(V(i,B_{\ell,r}))$ for all relevant indices $i$ and 
	the values $\sigma(H(j,B_{\ell,r}))$ for all relevant $j$ 
	using Lemma~\ref{le:block}. This takes $O((r-\ell) \log n)$ time.
	The pairs $(\ell,r)$ of $\I$ can also be interpreted as intervals.
	Since intervals of $\I$ with the same length are disjoint and there are $O(\log n)$
	different possible lengths of intervals in $\I$, the computation over all intervals of $\I$
	and all indices takes $O(n \log^2 n)$ time.
	
	For each $(\ell,r)\in \I$ we compute some additional prefix sums of columns
	and of rows, as follows.
	Assume that $B_{\ell,r}=B(i_0,i_1,j_0,j_1)$.
	For each $i$ with $i_0\le i\le i_1$, we define 
	$V_{\le}(i,B_{\ell,r})= \bigcup_{i_0\le i'\le i} V(i',B_{\ell,r})$.
	For each $j$ with $j_0\le j\le j_1$, we define 
	$H_{\le}(j,B_{\ell,r})= \bigcup_{j_0\le j'< j} H(j',B_{\ell,r})$.
	Using that 
	\begin{align*}
		\sigma(V_{\le}(i,B_{\ell,r}))~&=~\sigma(V_{\le}(i-1,B_{\ell,r}))+\sigma(V(i,B_{\ell,r})) 
				~~\text{for all $i$ with $i_0< i\le i_1$,}
	\end{align*}
	and a similar relation for $\sigma(H_{\le}(j,B_{\ell,r}))$,
	we can compute the values $\sigma(V_{\le}(i,B_{\ell,r}))$ and $\sigma(H_{\le}(j,B_{\ell,r}))$ 
	for all $i$ and $j$ in $O(\ell-r)$ time. In total, we spend $O(n\log n)$ time
	over all pairs $(\ell,r)$ of $\I$.
	
	Consider now a point $p_a$ of $P$. We can express the rectangle $R_{p_a}$
	as the union of $O(\log n)$ rectangles for which we have computed the relevant
	partial sums. See Figure~\ref{fig:ortho_hull5} for the intuition.
	More precisely, let $\I(a)$ be the pairs $(\ell,r)$ of $\I$ with $\ell< a < r$. 
	For each pair $(\ell,r)$ of $\I(a)$, we set 
	\[
		X(a,B_{\ell,r}) ~=~ \begin{cases}
				V_{\le}(a,B_{\ell,r}) & \text{if $a\le m(\ell,r)$,}\\
				H_{\le}(n-a+1,B_{\ell,r}) & \text{if $a> m(\ell,r)$.}
			\end{cases}
	\]
	Then $R_{p_a}$ is the (the closure of the) disjoint union of the cells in $X(a,B_{\ell,r})$, 
	where $(\ell,r)$ iterates over $\I(a)$.
	It follows that 
	\[
		\sum_{(\ell,r)\in \I(a)} \sigma(X(a,B_{(\ell,r)})) 
			~=~ \sigma(\{ c\in \A\mid c \subset R_{p_a} \} ~=~
		\sum_{c\in \A,~ c \subset R_{p_a}} \frac{\area(c)}{\NE(c)} ~=~ \phi(p_a,\V{ar}),
	\]
	where in the last equality we have used Lemma~\ref{le:percell}.
	\begin{figure}
		\centering
		\includegraphics[page=5]{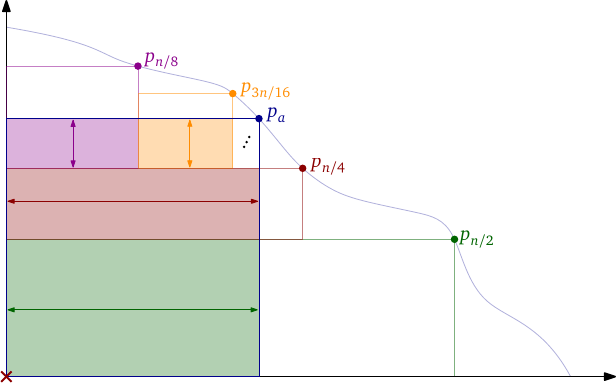}
		\caption{Expressing $R_{p_a}$ as the union of $O(\log n)$ rectangles
			of the form $X(a,B_{\ell,r})$. In this example,
			$3n/16 < a< n/4$.}
		\label{fig:ortho_hull5}
	\end{figure}
	
	For each point $p_a$ of $P$, we can compute the set $\I(a)$ 
	and then the sum $\sum_{(\ell,r)\in \I(a)} \sigma(X(a,B_{\ell,r}))
	=\phi(p_a,\V{ar})$ in $O(\log n)$ time.
	Over all points of $P$ we spend $O(n\log n)$ time in this last computation.
\end{proof}

When the point set is a chain over different quadrants, we can reduce
it to a few problems over the positive quadrant and obtain the following.
\begin{theorem}
\label{thm:ar_chain}
	If $P$ is a chain with $n$ points, then we can compute the Shapley values of 
	the \AreaAnchoredRectangles game in $O(n\log^2 n)$ time.
\end{theorem}
\begin{proof}
	If the point set $P$ is a chain, then we get a chain in each quadrant.
	As pointed out before (recall Figure\ref{fig:reduction1}), 
	we can treat each quadrant independently.
	Using reflections, we can transform the instance in any quadrant to the positive quadrant.
	(Note this may change the increasing/decreasing character of the chains.
	For example, an increasing chain in the northwest quadrant gets
	reflected into a decreasing chain in the positive quadrant.)
	Because of Lemma~\ref{le:ar_chain1} and Lemma~\ref{le:ar_chain2},
	we can compute the Shapley values of the \AreaAnchoredRectangles game for 
	each quadrant in $O(n\log^2 n)$ time.
	The result follows. 
\end{proof}

\subsection{General point sets}
\label{sec:ar_general}
We consider now the general case; thus the points do not form a chain.
See Figure~\ref{fig:ortho_hull1} for an example.
Like before, we restrict the discussion to the case where $P$ is in the positive quadrant.
In this case, $\NE(c_{i,j})$ is not a simple expression of $i,j$ anymore,
As mentioned in Section~\ref{sec:blocks}, we preprocess $P$ in $O(n\log n)$ time
such that $\NE(q)$ can be computed in $O(\log n)$ time~\cite{Willard85}.

In this scenario we consider horizontal bands.
A horizontal \DEF{band} $B$ is the block between two horizontal lines.
Thus, $B=\{ c_{i,j}\mid j_0\le j \le j_1\}$ for some indices $1\le j_0 \le j_1\le n$.
See Figure~\ref{fig:ortho_hull8}.
We keep using the notation introduced for blocks.
Thus, for each $i\in [n]$, let $V(i,B)$ be the vertical slab with the cells $c_{i,j}\in B$.
Let $P_B$ be the points of $P$ that are the top-right corner of some cell of $B$.
We use $k_B=|P_B|$. Because of our assumption on general position,
$k_B=j_1-j_0+1$ and thus $k_B$ is precisely the number of horizontal slabs in the band $B$.
Furthermore, for each point $p\in P_B$ we define the rectangle $R(p,B)$ as the cells of $B$
to the left and bottom of $p$. Formally $R(p,B)=\{ c\in B\mid c\subset R_p\}$.

\begin{figure}
	\centering
	\includegraphics[width=\textwidth,page=8]{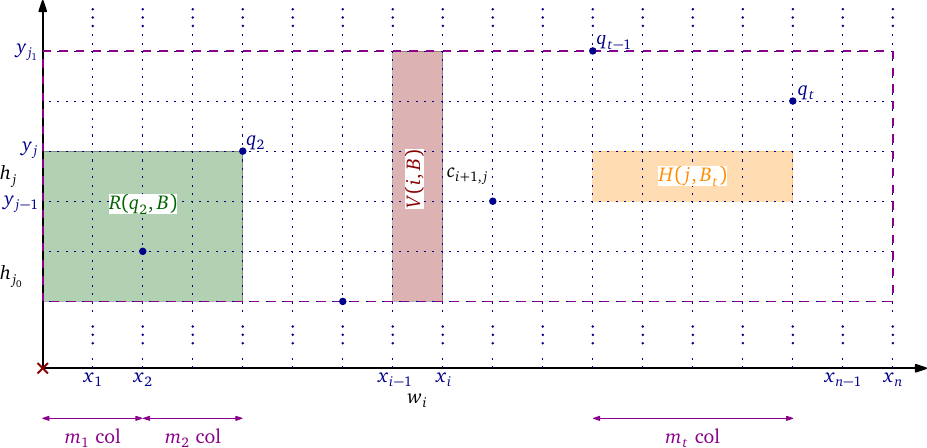}
	\caption{Notation for a band $B$.}
	\label{fig:ortho_hull8}
\end{figure}
	
\begin{lemma}
\label{le:ar_band}
	For a band $B$ with $k_B$ rows we can compute in $O(((k_B)^2+ n)\log n)$ time
	$\sigma(V(i,B))$, for all $i\in [n]$, and
	$\sigma(R(p,B))$, for all $p\in P_B$.	
\end{lemma}
\begin{proof}
	Assume that $B$ is defined by the row indices $j_0\le j_1$.
	Let $q_1,\dots, q_{k_B}$ be the points of $P_B$ sorted by increasing $x$ coordinate.
	Take also $q_0$ as a point on the $y$-axis and $q_{k_B+1}$ as a point on the right boundary.
	For each $t$ denote by $m_t$ the number of vertical slabs between the vertical lines through 
	$q_{t-1}$ and $q_t$.
	See Figure~\ref{fig:ortho_hull8}.
	
	We divide the band $B$ into blocks $B_1,B_2,\dots$ 
	using the vertical lines though the points of $P_B$.
	We get $k_B+1$ blocks (or $k_B$ if the rightmost point of $P$ belongs to $P_B$),
	and each of them is empty by construction.
	Since the block $B_t$ has $m_t$ vertical slabs and $k_B$ horizontal slabs
	we can compute $\sigma(V(\cdot,B_t))$ and $\sigma(H(\cdot,B_t))$ for all 
	the slabs within $B_t$ in $O((k_B + m_t) \log n)$ time 
	using Lemma~\ref{le:block}.
	Using that the $O(k_B)$ blocks $B_1,B_2,\dots$ are pairwise disjoint,
	which means that $\sum_t m_t = n$,
	we conclude that in   
	$O(((k_B)^2+ n)\log n)$ time
	we can compute all the values $\sigma(V(\cdot,B_t))$ 
	and $\sigma(H(\cdot,B_t))$ for all slabs within all blocks $B_t$.
	
	Since for each vertical slab $V(i,B)$ of $B$ there is one block $B_{t(i)}$
	that covers it, we then have $\sigma(V(i,B))=\sigma(V(i,B_{t(i)}))$ for such index $t(i)$.
	This means that we have already computed the values $\sigma(V(i,B))$ for all $i$.
	
	Now we explain the computation of $\sigma(R(p,B))$ for all $p\in P_B$.
	Note that within each block $B_t$ we have computed $\sigma(H(j,B_t))$ for all $j$.
	With this information we can compute the values $\sigma(R(p,B))$.
	Namely, for each block $B_t$ and each $j$ with $j_0\le j\le j_1$,
	we define 
	\[
		H_\le(j,B_t)) ~=~ \bigcup_{j_0 \le j'\le j}~~ \bigcup_{1\le t' \le t} H(j',B_{t'}).
	\]
	Using that 
	\begin{align*}
		 \forall j \text{ with }j_0 < j \le j_1:& ~~~ \sigma\left(H_\le(j,B_t)\right) ~=~
			\sigma( H_\le(j-1,B_t)) + \sigma(H(j,B_t)),\\
		\forall t \text{ with } 1< t\le k_B+1:& ~~~ \sigma\left(H_\le(j_0,B_t)\right) ~=~
			\sigma( H_\le(j_0,B_{t-1})) + \sigma(H(j_0,B_t)),
	\end{align*}
	we compute all the values $\sigma(H_\le(j,B_t))$ in $O((k_B)^2)$ time.
	Since $R(p,B)=H_\le(j(t),B_{t(p}))$ for some indices $j(p)$ and $t(p)$ 
	that we can easily obtain, the lemma follows.
\end{proof}

\begin{theorem}
\label{thm:ar_general}
	The Shapley values of the \AreaAnchoredRectangles game for $n$ points can be computed 
	in $O(n^{3/2} \log n)$ time.	
\end{theorem}
\begin{proof}
	As discussed earlier, it is enough to consider the case when $P$ is in the positive quadrant
	because the other quadrants can be transformed to this one.
	
	We split the set of cells into $k=\lceil \sqrt{n}\rceil$ bands, each with at most $k$
	horizontal slabs; one of the bands can have fewer slabs.
	For each band $B$ we use Lemma~\ref{le:ar_band}. This takes 
	$O(k \cdot (k^2+ n)\log n)= O(n^{3/2} \log n)$ time.
	For each band $B$ and $i\in [n]$, we further consider
	$V_\le(i,B)=\bigcup_{i'\le i} V(i,B)$ and compute the values
	$\sigma(V_\le(i,B))=\sum_{i'\le i} \sigma(V(i',B))$ using prefix sums.
	This takes $O(n^{3/2})$ additional time.

	\begin{figure}
		\centering
		\includegraphics[page=10]{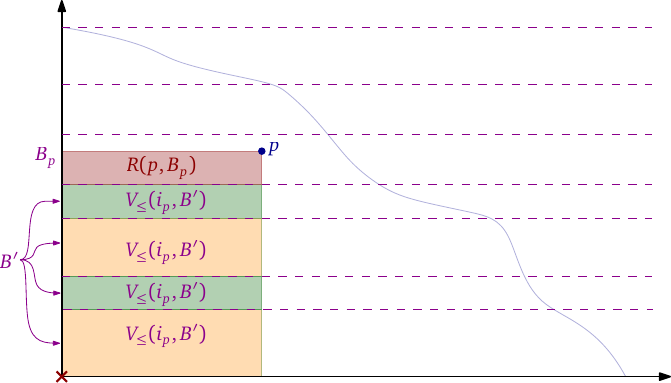}
		\caption{Expressing $R_p$ as the union of other groups of cells.}
		\label{fig:ortho_hull10}
	\end{figure}
	
	The cells inside a rectangle $R_p$ now correspond to the disjoint union
	\[
		R(p,B_p)\cup \bigcup_{B' \text{ below }  B_p} V_\le(i_p,B'),
	\]
	where $B_p$ is the band that contains $p$ (perhaps on its top boundary)
	and the index $i_p$ is such that the cell $c_{i_p,j}$ has $p$ on its top right corner.
	See Figure~\ref{fig:ortho_hull10}.
	Because of Lemma~\ref{le:percell}
	it follows that
	\[
		\phi(p,\V{ar}) ~=~ \sum_{c\in \A,~ c \subset R_{p_a}} \frac{\area(c)}{\NE(c)} ~=~
			\sigma(R(p,B_p)) + \sum_{B' \text{ below } B_p} \sigma(V_\le(i_p,B')).
	\]
	Since the relevant values are already computed, and there are $O(\sqrt{n})$ of them,
	namely one per band,
	we obtain $\phi(p,\V{ar})$ in $O(\sqrt{n})$ time per point $p\in P$.
\end{proof}

\section{Area of the bounding box}
\label{sec:boundingbox}

In this section we are interested in the \AreaBoundingBox game
defined by the characteristic function $\V{bb}$.
The structure of this section is similar to the structure of 
Section~\ref{sec:anchoredrectangles}. In particular, 
we keep assuming that all points have different coordinates and
we keep using the notation introduced in Section~\ref{sec:notation}.

First we note that it is enough to consider the \AreaAnchoredBoundingBox
problem and assume that the points are in one quadrant.
Recall that in this problem the origin $o$ has to be included in the bounding box.
Thus, it uses the characteristic function $\V{abb}(Q)=\area(\bb(Q\cup \{ o\}))$.

\begin{lemma}
\label{le:reductionbb}
	If we can solve the \AreaAnchoredBoundingBox problem for $n$ points
	in the first quadrant in time $T(n)$, then we can solve the \AreaBoundingBox game
	in $T(n)+O(n\log n)$ time.
	Furthermore, if we can solve the \AreaAnchoredBoundingBox problem for any
	$n$ points in a quadrant that form any chain in $T_c(n)$ time, 
	then we can solve the \AreaBoundingBox game for $n$ points that form a chain
	in $T_c(n)+O(n\log n)$ time.
\end{lemma}
\begin{proof}
	We use inclusion exclusion, as indicated in Figure~\ref{fig:reduction2}.
	There are characteristic functions $v_1,\dots,v_9$ such that,
	for each $Q\subseteq P$, we have $\V{bb}(Q)= \sum_{i=1}^9 v_i(Q)$.
	Moreover, each characteristic function $v_i$ is either a constant value game 
	(first summand in Figure~\ref{fig:reduction2}),
	isometrically equivalent to an \AreaBand game (last four terms 
	in Figure~\ref{fig:reduction2}), or isometrically equivalent to 
	an \AreaAnchoredBoundingBox problem with all points in one quadrant (the remaining
	four summands in Figure~\ref{fig:reduction2}). 
	Finally, note that if the points form a chain, they also form a chain in 
	each of the cases possibly exchanging the increasing/decreasing character.
	\begin{figure}
		\centering
		\includegraphics[width=\textwidth,page=1]{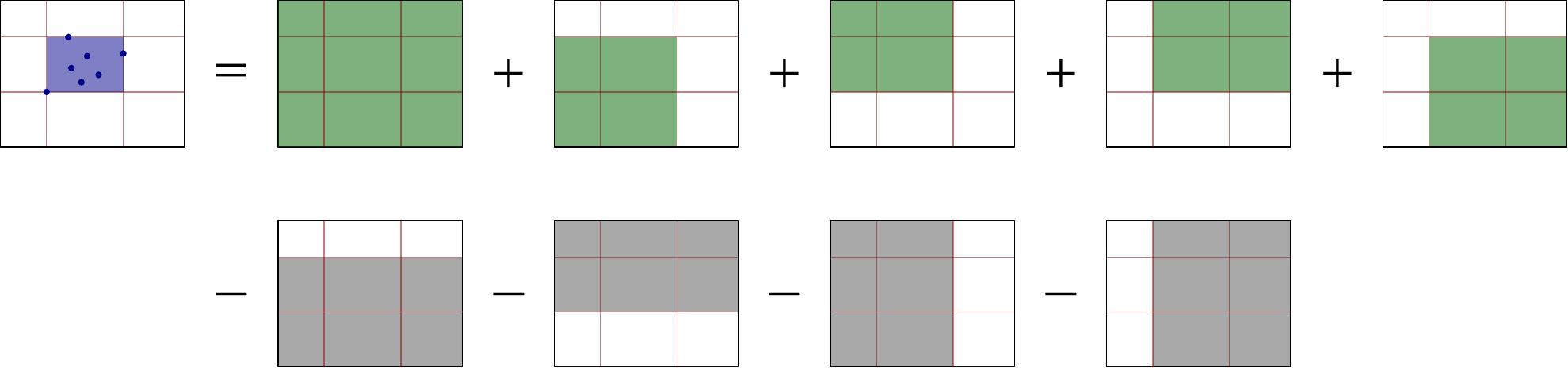}
		\caption{The \AreaBoundingBox game is a sum of other games, where the 
			non-trivial games are \AreaAnchoredBoundingBox games with points in one quadrant.}
		\label{fig:reduction2}
	\end{figure}
\end{proof}

\subsection{Interpreting Shapley values geometrically}

First, we reduce the problem of computing Shapley values of 
the \AreaAnchoredBoundingBox to a purely geometric problem. 
The situation here is slightly more complicated than in Section~\ref{sec:anchoredrectangles}
because a cell $c$ of $\A$ is inside $\bb(Q\cup\{ o\})$ if
and only if $Q$ contains some point in $\SNE(c)$ or
it contains some point in $\SNW(c)$ and in $\SSE(c)$.

For a cell $c$ of $\A$ we define the following values
\begin{align*}
	\psi_{\SNE}(c) ~&=~ \frac{1}{\NE(c)+\NW(c)}+ \frac{1}{\NE(c)+\SE(c)}-
					\frac{1}{\NE(c)+\NW(c)+\SE(c)},\\
	\psi_{\SNW}(c) ~&=~ \frac{1}{\NE(c)+\NW(c)} - \frac{1}{\NE(c)+\NW(c)+\SE(c)},\\
	\psi_{\SSE}(c) ~&=~ \frac{1}{\NE(c)+\SE(c)} - \frac{1}{\NE(c)+\NW(c)+\SE(c)}.
\end{align*}

The following lemma is summarized in Figure~\ref{fig:boundingbox1}.

\begin{figure}
	\centering
	\includegraphics[page=1]{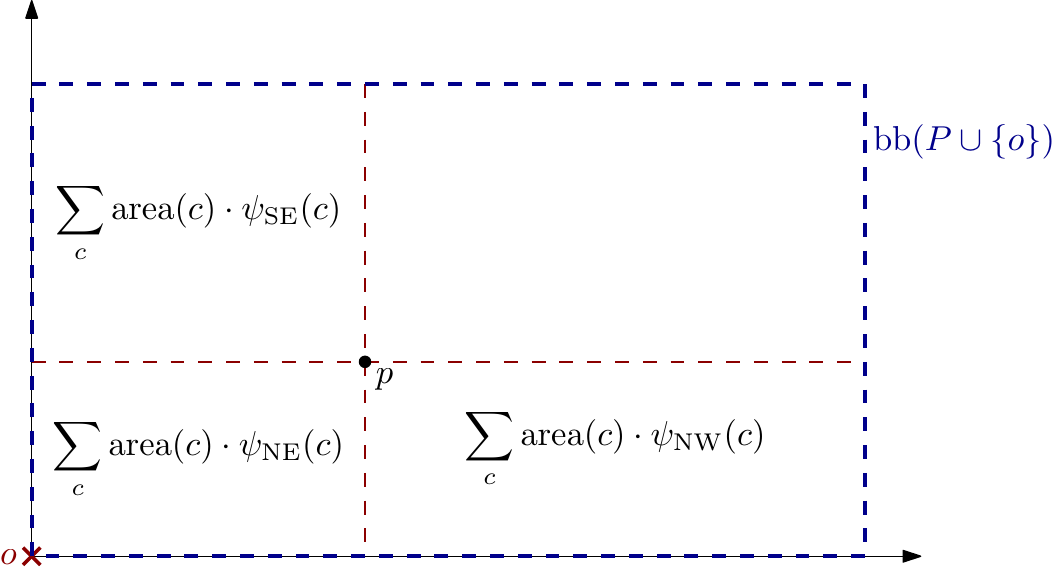}
	\caption{The formula in Lemma~\ref{le:percellbb}.}
	\label{fig:boundingbox1}
\end{figure}
	
\begin{lemma}
\label{le:percellbb}
	If $P$ is in the positive quadrant, then for each $p\in P$ the Shapley value
	$\phi(p,\V{abb})$ is
	\begin{align*}
		  \sum_{c\in \A,~ p\in \SNE(c)} \area(c)\cdot \psi_{\SNE}(c)  
		+ \sum_{c\in \A,~ p\in \SNW(c)} \area(c)\cdot \psi_{\SNW}(c) 
		+ \sum_{c\in \A,~ p\in \SSE(c)} \area(c)\cdot \psi_{\SSE}(c). 
	\end{align*}
\end{lemma}
\begin{proof}
	Each cell $c$ of the arrangement $\A$ defines a game for the set of players $P$ 
	with characteristic function
	\[
		v_c(Q) ~=~\begin{cases}
			\area(c) & \text{if $c\subset \bb(Q\cup \{ o \})$,}\\
			0 & \text{otherwise.}
			\end{cases}
	\]
	First note $\Delta(v_c,\pi,p)$ can only take the values $0$ and $\area(c)$.
	To analyze the Shapley values of the game defined by $v_c$ 
	we use Lemma~\ref{le:permutations1}.
	
	Consider a point $p\in \SNE(c)$.
	For a permutation $\pi \in \Pi(P)$, we have $\Delta(v_c,\pi,p)= \area(c)$
	if and only if $p$ is the first point of $\SNE(c)$ in the permutation $\pi$
	and all the points of $\SNW(c)$ or all the points of $\SSE(c)$
	are after $p$ in the permutation $\pi$.
	According to the first item of Lemma~\ref{le:permutations1},
	with $N=P$, $a=p$, $A=\SNE(c)$, $B=\SNW(c)$ and $C=\SSE(c)$,
	there are precisely
	\[
		n! \cdot \left(\frac{1}{\NE(c)+\NW(c)}+ \frac{1}{\NE(c)+\SE(c)}-
				\frac{1}{\NE(c)+\NW(c)+\SE(c)}\right) ~=~
		n!\cdot \psi_{\SNE}(c)
	\]
	permutations that fulfill this criteria.
	This means that, for each $p\in \SNE(c)$ we have
	\[
		\phi(p,v_c) ~=~ \area(c) \cdot \psi_{\SNE}(c).
	\]
	
	Consider now a point $p\in \SNW(c)$.
	For a permutation $\pi \in \Pi(P)$, we have $\Delta(v_c,\pi,p)= \area(c)$
	if and only if $p$ is the first point of $\SNW(c)$ in the permutation $\pi$,
	all the points of $\SNE(c)$ are after $p$ in the permutation $\pi$,
	and at least one point of $\SSE(c)$ is before $p$ in the permutation $\pi$.
	According to the second item of Lemma~\ref{le:permutations1},
	with $N=P$, $b=p$, $B=\SNW(c)$, $A=\SNE(c)$ and $C=\SSE(c)$,
	there are precisely
	\[
		n! \cdot \left( \frac{1}{\NE(c)+\NW(c)} - \frac{1}{\NE(c)+\NW(c)+\SE(c)} \right) ~=~
		n! \cdot \psi_{\SNW}(c)
	\]
	permutations that fulfill this criteria.
	This means that, for each $p\in \SNW(c)$ we have
	\[
		\phi(p,v_c) ~=~ \area(c) \cdot \psi_{\SNW}(c).
	\]
	
	The case for a point $p\in \SSE(c)$ is symmetric to the case $p\in \SNW(c)$,
	where the roles of $\SE(c)$ and $\NW(c)$ are exchanged.
	Therefore we conclude
	\[
		\phi(p,v_c) ~=~ \area(c)\cdot 
			\begin{cases}
				\psi_{\SNE}(c)  & \text{if $p\in \SNE(c)$,}\\
				\psi_{\SNW}(c)	& \text{if $p\in \SNW(c)$,}\\
				\psi_{\SSE}(c)	& \text{if $p\in \SSE(c)$.}								
			\end{cases}
	\]
	As a sanity check it good to check that
	\begin{align*}
		\NE(c)\cdot \psi_{\SNE}(c) + \NW(c)\cdot \psi_{\SNW}(c) +
		\SE(c)\cdot \psi_{\SSE}(c) ~=~ 1.
	\end{align*}
	This is the case.
	
	Noting that we have, 
	\[
		\V{abb}(Q) ~=~ \sum_{c\in \A} v_c(Q)~~\text{ for all $Q\subset P$,}
	\]
	the result follows from the linearity of Shapley values.
\end{proof}

For each subset of cells $C$ of $\A$, we define 
\begin{align*}
	\sigma_{\SNE}(C) ~&=~ \sum_{c\in C} \area(c)\cdot \psi_{\SNE}(c),\\ 
	\sigma_{\SNW}(C) ~&=~ \sum_{c\in C} \area(c)\cdot \psi_{\SNW}(c),\\ 
	\sigma_{\SSE}(C) ~&=~ \sum_{c\in C} \area(c)\cdot \psi_{\SSE}(c).\\ 
\end{align*}
Like in the case of anchored boxes, our objective here
is to compute these values for several vertical and horizontal slabs.

\subsection{Handling empty blocks}

In the following we assume that we have preprocessed $P$ in $O(n\log n)$ time
such that $\NE(q)$, $\NW(q)$ and $\SE(q)$ can be computed in $O(\log n)$ time
for each point $q$ given at query time~\cite{Willard85}.
We use again multipoint evaluation to compute the values 
$\sigma_*(\cdot)$ for each vertical and
horizontal slab of an empty block and for each $*\in \{\SNE,\SNW,\SSE\}$.
However, the treatment has to be a bit more careful now because we have
to deal with different rational functions.

\begin{lemma}
\label{le:blockBB}
	Let $B$ be an empty block with $k$ columns and rows.
	We can compute in $O(k \log n)$ time the 
	values $\sigma_{*}(C)$ for all slabs $C$ within $B$ and each $*\in \{\SNE,\SNW,\SSE\}$.
\end{lemma}
\begin{proof}
	The proof is very similar to the proof of Lemma~\ref{le:block},
	but some adaptation has to be made.

	Assume that $B$ is the block $B(i_0,i_1,j_0,j_1)$.
	We explain how to compute the values $\sigma_{\SNE}(V(i,B))$ for all $i_0\le i\le i_1$.
	The technique to compute $\sigma_{\SNE}(H(j_0,B)),\dots, \sigma_{\SNE}(H(j_1,B))$
	and for $\SNW$ and $\SSE$ instead of $\SNE$ is similar.
	
	For each $\ell$ we define
	$J(\ell)=\{ j\mid j_0\le j \le j_1, ~ \NE(c_{i_0,j})+\SE(c_{i_0,j})=\ell\}$
	and
	$J'(\ell)=\{ j\mid j_0\le j \le j_1, ~ \NE(c_{i_0,j})+\NW(c_{i_0,j})+\SE(c_{i_0,j})=\ell\}$.
	Let $\ell_0$ and $\ell_1$ be the minimum and the maximum $\ell$ such that $J(\ell)\neq 0$, respectively.
	Similarly, let $\ell'_0$ and $\ell'_1$ be the minimum and the maximum $\ell$ such that $J'(\ell)\neq 0$.
	
	We set up the following fractional constant and rational functions with variable $x$:
	\begin{align*}
		R_1    ~&=~ \sum_{j=j_0}^{j_1} \frac{h_j}{\NE(c_{i_0,j})+\NW(c_{i_0,j})},\\
		R_2(x) ~&=~ \sum_{\ell=\ell_0}^{\ell_1} \frac{~\sum_{j\in J(\ell)}h_j~}{\ell+x},\\
		R_3(x) ~&=~ \sum_{\ell=\ell'_0}^{\ell'_1} \frac{~\sum_{j\in J'(\ell)}h_j~}{\ell+x}.
	\end{align*}
	Each of them is the sum of at most $k$ fractions. 
	As discussed in the proof of Lemma~\ref{le:block},
	$R_2(x)$ and $R_3(x)$ can be rewritten so that we can use Lemma~\ref{le:multipoint_evaluation}
	to evaluate them at consecutive points.
	The coefficients can be computed in $O(k\log n)$ time.	
	
	\begin{figure}
		\centering
		\includegraphics[width=\textwidth,page=2]{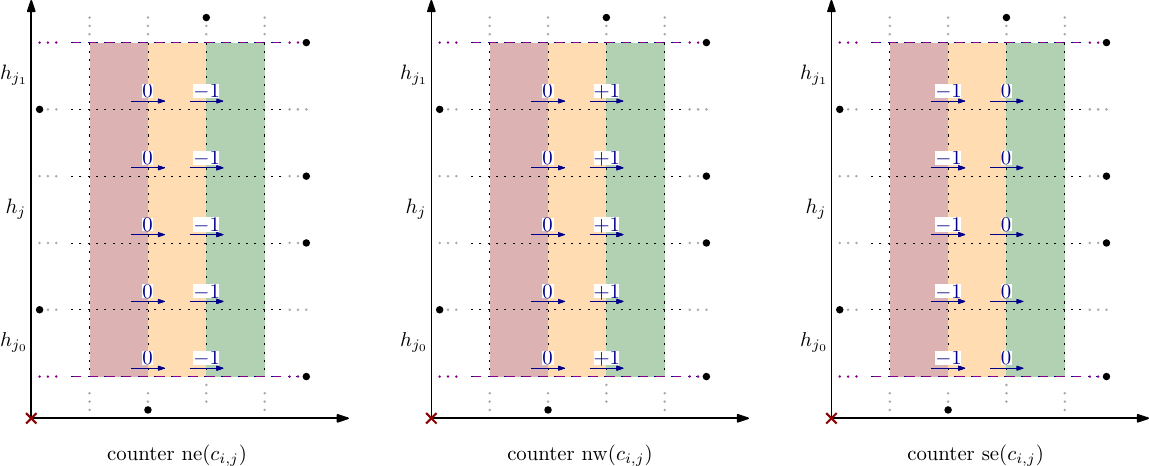}
		\caption{Changes in the counter $\NE(\cdot)$ (left), $\NW(\cdot)$ (center) 
			and $\SE(\cdot)$ (right) depending on the position of the points of $P$.}
		\label{fig:boundingbox2}
	\end{figure}

	Consider two consecutive vertical slabs $V(i,B)$ and $V(i+1,B)$ within  
	the block $B$. Figure~\ref{fig:boundingbox2} shows the changes in the counters.
	Because the block $B$ is empty we have the following properties:
	\begin{itemize}
		\item The sum $\NE(c_{i,j})+\NW(c_{i,j})$ is independent of $i$.
			Thus we have 
			\[
				\NE(c_{i,j})+\NW(c_{i,j})= \NE(c_{i_0,j})+\NW(c_{i_0,j})
			\]
			for all relevant $i$ and $j$.
		\item The difference 
			\[ 
				\bigl(\NE(c_{i+1,j})+\SE(c_{i+1,j})\bigr) - \bigl(\NE(c_{i,j})+\SE(c_{i,j})\bigr)
			\]
			is always $-1$.
			Therefore we have
			\[
				\NE(c_{i,j})+\SE(c_{i,j})= \NE(c_{i_0,j})+\SE(c_{i_0,j}) - (i-i_0)
			\]
			for all relevant $i$ and $j$.
			In particular, for each $j\in J(\ell)$ we have $\NE(c_{i,j})+\SE(c_{i,j})= \ell -(i-i_0)$.
		\item The difference 
			\[ 
				\bigl(\NE(c_{i+1,j})+\NW(c_{i+1,j})+\SE(c_{i+1,j})\bigr) - 
				\bigl(\NE(c_{i,j})+\NE(c_{i,j})+\SE(c_{i,j})\bigr)
			\]
			depends on whether the point with $x(p)=x_i$ is above or below $B$.
			Thus, this difference is independent of the row $j$.
			This means that there exist, for each index $i$, a value $\delta_i$
			such that 
			\[
				\NE(c_{i,j})+\NW(c_{i,j})+\SE(c_{i,j})= 
				\NE(c_{i_0,j})+\NW(c_{i_0,j})+ \SE(c_{i_0,j}) + \delta_i
			\]
			for all relevant $j$.
			In particular, for each $j\in J'(\ell)$ we have $\NE(c_{i,j})+\NW(c_{i,j})+\SE(c_{i,j})= \ell +\delta_i$.
			Each single value $\delta_i$ can be computed as $\delta_i=\SE(c_{i,j})- \SE(c_{i_0,j})$
			in $O(\log n)$ time using range counting queries.
	\end{itemize}

	Note that for each relevant $i$ we have
	\begin{align*}
		w_i&\cdot (R_1+R_2(-i+i_0)-R_3(\delta_i))\\
			&=~ w_i \cdot \left(\sum_{j=j_0}^{j_1} \frac{h_j}{\NE(c_{i_0,j})+\NW(c_{i_0,j})} ~+~
				\sum_{\ell=\ell_0}^{\ell_1} \frac{~\sum_{j\in J(\ell)}h_j~}{\ell-i+i_0} ~+~
				\sum_{\ell=\ell'_0}^{\ell'_1} \frac{~\sum_{j\in J'(\ell)}h_j~}{\ell+\delta_i}
				\right)\\
			&=~ \sum_{j=j_0}^{j_1} \left( \frac{w_i h_j}{\NE(c_{i,j})+\NW(c_{i,j})} +
					\frac{w_i h_j}{\NE(c_{i,j})+\SE(c_{i,j})}  
					-\frac{w_i h_j}{\NE(c_{i,j})+\NW(c_{i,j})+\SE(c_{i,j})}\right) \\
			&=~ \sum_{j=j_0}^{j_1} \area(c_{i,j}) \cdot \psi_{\SNE}(c_{i,j})\\
			&=~ \sigma_{\SNE}(V(i,B)).
	\end{align*}
	
	Thus, we need to evaluate $R_2(x)$ at the values $x=0,-1,\dots,i_0-i_1$
	and we have to evaluate $R_3(x)$ at the values $\delta_i$ for $i=i_0,\dots,i_1$.
	According to Lemma~\ref{le:multipoint_evaluation}, this can be done 
	in $O(k\log n)$ time.  
	After this, we get each value $\sigma_{\SNE}(V(i,B))$ in constant time. 
\end{proof}

\subsection{Algorithms for bounding box}
The same techniques that were used in Sections~\ref{sec:ar_chain} and~\ref{sec:ar_general}
work in this case.
We compute partial sums $\sigma_{*}(\cdot)$ for several vertical
bands and horizontal bands. The only difference is that we have to cover the whole
bounding box $\bb(P\cup \{ o\})$, instead of only the portion that was dominated
by some point. 
In the case of an increasing chain in the positive quadrant, the 
anchored bounding box and the union of anchored rectangles is actually the
same object.
In the case of a decreasing chain, 
we have to work above and below the chain to cover the whole bounding box. 
See Figure~\ref{fig:boundingbox3}. 
This is twice as much work, so it does not affect
the asymptotic running time.
For arbitrary point sets, we again use the bands with roughly $\sqrt{n}$ horizontal slabs,
and break each slab into empty boxes, as it was done in Lemma~\ref{le:ar_band}.
For the partial sums that we consider (like $\sigma(H_\le(\cdot))$ and $\sigma(V_\le(\cdot))$)
we also construct the symmetric versions $\sigma_*(H_\ge(\cdot))$ and $\sigma_*(V_\ge(\cdot))$.
With this information we can recover the sums that we need in each of the
three quadrants with an apex in $p\in P$; recall Figure~\ref{fig:boundingbox1}.
In the case of a decreasing chain, we get an extra case that is handled by
noting that
\[
	\eta_a ~=~ \sum_{c\in \A,~ p_a\in \SSE(c)} \area(c)\cdot \psi_{\SSE}(c)
\]
is
\[
	\eta_a ~=~ \eta_{a-1} + 
	\sum_{j=a+1}^{n} \area(c_{a,j})\cdot \psi_{\SSE}(c_{a,j}) + 
	\sum_{i=1}^{a-1} \area(c_{i,n-a+1})\cdot \psi_{\SSE}(c_{i,n-a+1}).
\]
The last two terms are a vertical and a horizontal band around the cell $c_{a,n-a+1}$.
The sums for $\psi_{\SNW}$ are handled similarly.
We omit the details and summarize.

\begin{figure}
	\centering
	\includegraphics[page=3]{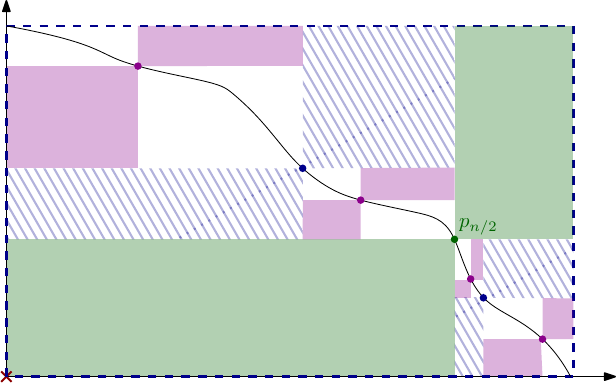}
	\caption{The empty blocks to be considered for the decreasing chain.}
	\label{fig:boundingbox3}
\end{figure}

\begin{theorem}
\label{thm:abb}
	The Shapley values of the \AreaAnchoredBoundingBox game for $n$ points can be computed 
	in $O(n^{3/2} \log n)$ time. If the points form a chain,
	then we need $O(n \log^2 n)$ time.
\end{theorem}

Because of Lemma~\ref{le:reductionbb} we also obtain the following.
\begin{theorem}
\label{thm:bb}
	The Shapley values of the \AreaBoundingBox game for $n$ points can be computed 
	in $O(n^{3/2} \log n)$ time. If the points form a chain,
	then we need $O(n \log^2 n)$ time.
\end{theorem}

\section{Conclusions}
\label{sec:conclusions}

We have discussed the efficient computation of Shapley values,
a classical topic in game theory, for coalitional games defined for points in the plane
and characteristic functions given by the area of geometric objects.
For axis-parallel problems we used computer algebra to get
faster algorithms. For non-axis-parallel problems we provided
efficient algorithms based on decomposing the sum into parts
and grouping permutations that contribute equally to a part.

In game theory, quite often one considers coalitional games 
where some point, say the origin, has to be included in the solutions
that are considered. For example, we could use
the minimum enclosing disk that contains also the origin.
We do this for the anchored versions of the games. 
This setting is meaningful in several scenarios, for example,
when we split the costs to connect to something.
Our results also hold in this setting through an easy adaptation.

The problems we consider here are a new type of stochastic problems
in computational geometry. The relation to other problems in 
stochastic computational geometry is unclear. For example, 
computing the expected length of the minimum spanning tree (MST) in the plane 
for a stochastic point set is \#P-hard~\cite{KamousiCS11}.
Does this imply the same for the coalitional game based on the length
of the Euclidean MST? 
Is there a two-way relation between Shapley values and 
the stochastic version for geometric problems?
In particular, is there a FPRAS for computing
the Shapley values of the game defined using the Euclidean MST? 
For the stochastic model this is shown by Kamousi et al.~\cite{KamousiCS11}. 
Note that the length of the Euclidean spanning tree is not monotone: adding
points may reduce its length. Thus, some Shapley values could
be potentially $0$, which, at least intuitively, 
makes it harder to get approximations.

Finally, it would be worth to understand whether there is some relation
between Shapley values in geometric settings and the concept of depth in point sets,
in particular for the Tukey depth. For stochastic points, this relation
has been explored and exploited by Agarwal et al.~\cite{AgarwalHSYZ17}.

Our algorithms generalize to higher dimensions, increasing the degree
of the polynomial describing the running time. Are there substantial improvements
possible for higher dimensions?
In dimension $d$, computing the volume of the bounding box can be done in $O(dn)$ time,
while the obvious algorithms to compute the Shapley values for the corresponding
game require $n^{\Theta(d)}$ time. 
Is this linear dependency on $d$ in the degree of the polynomial actually needed, 
under some standard assumption in complexity theory?

\subsection*{Acknowledgments.}
The authors are very grateful to Sariel Har-Peled for fruitful discussions. In particular, 
the possible link to Tukey depth was pointed out by him.


\begin{thebibliography}{10}

\bibitem{AgarwalHSYZ17}
P.~K. Agarwal, S.~Har{-}Peled, S.~Suri, H.~Yildiz, and W.~Zhang.
\newblock Convex hulls under uncertainty.
\newblock {\em Algorithmica} 79(2):340--367, 2017,
  \url{https://doi.org/10.1007/s00453-016-0195-y}.

\bibitem{AgarwalKSS18}
P.~K. Agarwal, N.~Kumar, S.~Sintos, and S.~Suri.
\newblock Range-max queries on uncertain data.
\newblock {\em J. Comput. Syst. Sci.} 94:118--134, 2018,
  \url{https://doi.org/10.1016/j.jcss.2017.09.006}.

\bibitem{AjwaniRST07}
D.~Ajwani, S.~Ray, R.~Seidel, and H.~R. Tiwary.
\newblock On computing the centroid of the vertices of an arrangement and
  related problems.
\newblock {\em 10th International Workshop Algorithms and Data Structures,
  {WADS} 2007}, pp.~519--528. Springer, Lecture Notes in Computer Science 4619,
  2007, \url{https://doi.org/10.1007/978-3-540-73951-7_45}.

\bibitem{AronovK18}
B.~Aronov and M.~J. Katz.
\newblock Batched point location in {SINR} diagrams via algebraic tools.
\newblock {\em {ACM} Trans. Algorithms} 14(4):41:1--41:29, 2018,
  \url{http://doi.acm.org/10.1145/3209678}.

\bibitem{BergCKO08}
M.~de~Berg, O.~Cheong, M.~J. van Kreveld, and M.~H. Overmars.
\newblock {\em Computational geometry: algorithms and applications, 3rd
  Edition}.
\newblock Springer, 2008.

\bibitem{Deng2008}
X.~Deng and Q.~Fang.
\newblock Algorithmic cooperative game theory.
\newblock {\em Pareto Optimality, Game Theory And Equilibria}, pp.~159--185.
  Springer New York, 2008, \url{https://doi.org/10.1007/978-0-387-77247-9_7}.

\bibitem{DengP94}
X.~Deng and C.~H. Papadimitriou.
\newblock On the complexity of cooperative solution concepts.
\newblock {\em Mathematics of Operations Research} 19(2):257--266, 1994,
  \url{https://doi.org/10.1287/moor.19.2.257}.

\bibitem{EdelsbrunnerOS86}
H.~Edelsbrunner, J.~{O'Rourke}, and R.~Seidel.
\newblock Constructing arrangements of lines and hyperplanes with applications.
\newblock {\em {SIAM} J. Comput.} 15(2):341--363, 1986,
  \url{https://doi.org/10.1137/0215024}.

\bibitem{EdelsbrunnerSS93}
H.~Edelsbrunner, R.~Seidel, and M.~Sharir.
\newblock On the zone theorem for hyperplane arrangements.
\newblock {\em {SIAM} J. Comput.} 22(2):418--429, 1993,
  \url{https://doi.org/10.1137/0222031}.

\bibitem{Faigle1998}
U.~Faigle, S.~P. Fekete, W.~Hochst{\"a}ttler, and W.~Kern.
\newblock On approximately fair cost allocation in {E}uclidean {TSP} games.
\newblock {\em Operations-Research-Spektrum} 20(1):29--37, 1998,
  \url{https://doi.org/10.1007/BF01545526}.

\bibitem{Ferguson}
T.~S. Ferguson.
\newblock Game theory, 2nd edition, 2014.
\newblock Electronic text available at
  \url{https://www.math.ucla.edu/~tom/Game_Theory/Contents.html}.

\bibitem{FinkHKS17}
M.~Fink, J.~Hershberger, N.~Kumar, and S.~Suri.
\newblock Hyperplane separability and convexity of probabilistic point sets.
\newblock {\em JoCG} 8(2):32--57, 2017,
  \url{http://jocg.org/index.php/jocg/article/view/321}.

\bibitem{KamousiCS11}
P.~Kamousi, T.~M. Chan, and S.~Suri.
\newblock Stochastic minimum spanning trees in {E}uclidean spaces.
\newblock {\em Proceedings of the 27th {ACM} Symposium on Computational
  Geometry, SoCG'11}, pp.~65--74. {ACM}, 2011,
  \url{http://doi.acm.org/10.1145/1998196.1998206}.

\bibitem{KamousiCS14}
P.~Kamousi, T.~M. Chan, and S.~Suri.
\newblock Closest pair and the post office problem for stochastic points.
\newblock {\em Comput. Geom.} 47(2):214--223, 2014,
  \url{https://doi.org/10.1016/j.comgeo.2012.10.010}.

\bibitem{Littlechild}
S.~Littlechild and G.~Owen.
\newblock A simple expression for the {S}hapely value in a special case.
\newblock {\em Management Science} 20(3):370--372, 1973,
  \url{https://doi.org/10.1287/mnsc.20.3.370}.

\bibitem{Matousek02}
J.~Matou\v{s}ek.
\newblock {\em Lectures on Discrete Geometry}.
\newblock Springer-Verlag, Berlin, Heidelberg, 2002.

\bibitem{Megiddo78}
N.~Megiddo.
\newblock Computational complexity of the game theory approach to cost
  allocation for a tree.
\newblock {\em Mathematics of Operations Research} 3(3):189--196, 1978,
  \url{https://doi.org/10.1287/moor.3.3.189}.

\bibitem{MorozA16}
G.~Moroz and B.~Aronov.
\newblock Computing the distance between piecewise-linear bivariate functions.
\newblock {\em {ACM} Trans. Algorithms} 12(1):3:1--3:13, 2016,
  \href{http://dx.doi.org/10.1145/2847257}%
{doi:10.1145/2847257}, \url{http://doi.acm.org/10.1145/2847257}.

\bibitem{Myerson}
R.~B. Myerson.
\newblock {\em Game theory - Analysis of Conflict}.
\newblock Harvard University Press, 1997.

\bibitem{NisanRTV07}
N.~Nisan, T.~Roughgarden, {\'E}.~Tardos, and V.~V. Vazirani.
\newblock {\em Algorithmic Game Theory}.
\newblock Cambridge University Press, 2007.

\bibitem{OsborneRubinstein}
M.~J. Osborne and A.~Rubinstein.
\newblock {\em A course in game theory}.
\newblock The MIT Press, 1994.

\bibitem{Perez-Lantero16}
P.~P{\'{e}}rez{-}Lantero.
\newblock Area and perimeter of the convex hull of stochastic points.
\newblock {\em Comput. J.} 59(8):1144--1154, 2016,
  \url{https://doi.org/10.1093/comjnl/bxv124}.

\bibitem{PuertoTP11}
J.~Puerto, A.~Tamir, and F.~Perea.
\newblock A cooperative location game based on the 1-center location problem.
\newblock {\em European Journal of Operational Research} 214(2):317--330, 2011,
  \url{https://doi.org/10.1016/j.ejor.2011.04.020}.

\bibitem{PuertoTP12}
J.~Puerto, A.~Tamir, and F.~Perea.
\newblock Cooperative location games based on the minimum diameter spanning
  {S}teiner subgraph problem.
\newblock {\em Discrete Applied Mathematics} 160(7-8):970--979, 2012,
  \url{https://doi.org/10.1016/j.dam.2011.07.020}.

\bibitem{Roth88}
A.~E. Roth, editor.
\newblock {\em The {S}hapley Value: Essays in Honor of {L}loyd {S}. {S}hapley}.
\newblock Cambridge University Press, 1988.

\bibitem{Thomson13}
W.~Thomson.
\newblock Cost allocation and airport problems, 2013.
\newblock Rochester Center for Economic Research Working Paper. Version of 2014
  available at
  \url{http://www.iser.osaka-u.ac.jp/collabo/20140524/Airport_Problems.pdf}.

\bibitem{Welzl91}
E.~Welzl.
\newblock Smallest enclosing disks (balls and ellipsoids).
\newblock {\em New Results and New Trends in Computer Science}, pp.~359--370.
  Springer, Lecture Notes in Computer Science 555, 1991.

\bibitem{Willard85}
D.~E. Willard.
\newblock New data structures for orthogonal range queries.
\newblock {\em SIAM J. Comput.} 14:232--253, 1985,
  \url{https://doi.org/10.1137/0214019}.

\bibitem{Winter}
E.~Winter.
\newblock The {S}hapley value.
\newblock {\em Handbook of Game Theory with Economic Applications}, 1 edition,
  vol.~3, chapter~53, pp.~2025--2054. Elsevier, 2002,
  \url{https://doi.org/10.1016/S1574-0005(02)03016-3}.

\bibitem{XueLJ18}
J.~Xue, Y.~Li, and R.~Janardan.
\newblock On the separability of stochastic geometric objects, with
  applications.
\newblock {\em Comput. Geom.} 74:1--20, 2018,
  \href{http://dx.doi.org/10.1016/j.comgeo.2018.06.001}%
{doi:10.1016/j.comgeo.2018.06.001},
  \url{https://doi.org/10.1016/j.comgeo.2018.06.001}.

\end{thebibliography}
\end{document}